%% file: main.tex
\documentclass[11pt]{article}
\usepackage[margin=1in]{geometry}
\usepackage[utf8]{inputenc}
\usepackage{comment}
\usepackage{appendix}
\input{utils.tex}

\title{Adaptive Massively Parallel Constant-round Tree Contraction}
\date{}
\author{Anonymous Authors}
\author{
MohammadTaghi Hajiaghayi\footnote{Supported by the NSF BIGDATA Grant No. 1546108, NSF SPX Grant No. 1822738, and NSF AF Grant No. 2114269.}\\
University of Maryland\\
\texttt{hajiagha@cs.umd.edu}
\and
Marina Knittel\footnote{Supported by the NSF BIGDATA Grant No. 1546108, NSF SPX Grant No. 1822738, ARCS Endowment Award, and Ann G. Wylie Fellowship.}\\
University of Maryland\\
\texttt{mknittel@cs.umd.edu}
\and
Hamed Saleh\\
University of Maryland\\
\texttt{hamed@cs.umd.edu}
\and
Hsin-Hao Su\footnote{Supported by NSF Grant No. CCF-2008422.}\\
Boston College\\
\texttt{suhx@bc.edu}
}

\renewcommand{\varepsilon}{\epsilon}

\begin{document}

\maketitle

\begin{abstract}
Miller and Reif's FOCS'85~\cite{miller1989parallel} classic and fundamental {\em tree contraction} algorithm  is a broadly applicable technique for the parallel solution of a large number of tree problems. Additionally it is also used as an algorithmic design technique for a large number of parallel graph algorithms. In all previously explored models of computation, however, tree contractions have only been achieved in $\Omega(\log n)$ rounds of parallel run time. In this work, we not only introduce a generalized tree contraction method but also show it can be computed highly efficiently in $O(1/\epsilon^3)$ rounds in the Adaptive Massively Parallel Computing (AMPC) setting, where each machine has $O(n^\epsilon)$ local memory for some $0 < \epsilon < 1$. AMPC is a practical extension of Massively Parallel Computing (MPC) which utilizes distributed hash tables~\cite{bateni2017affinity,behnezhad2019massively, kiveris2014connected}. In general, MPC is an abstract model for MapReduce, Hadoop, Spark, and Flume which are currently widely used across industry and has been studied extensively in the theory community in recent years. Last but not least, we show that our results extend to multiple problems on trees, including but not limited to maximum and maximal matching, maximum and maximal independent set, tree isomorphism testing, and more.
\end{abstract}

\input{introduction}
\input{preliminaries}

\input{contributions}
\input{applications}


\bibliographystyle{alpha}
\bibliography{references}



\end{document}

%% file: utils.tex
\usepackage{amsfonts, amsmath, amssymb, amsthm, enumerate, gensymb, mathtools, tikz, graphicx, thmtools, hyperref, url, thm-restate, cite}

\usepackage[utf8]{inputenc}
\usepackage{graphicx}
\graphicspath{{images/}}
\usepackage{amsfonts}
\usepackage{amsthm}
\usepackage{xcolor}
\usepackage{pdftexcmds}
\usepackage{caption}
\usepackage{subcaption}
\newcommand{\marina}[1]{}
\newcommand{\hamed}[1]{}
\newcommand{\hsinhao}[1]{}
\newcommand{\blackbox}[1]{\textsc{#1}}

\usepackage[ruled,vlined,linesnumbered]{algorithm2e}

\newtheorem{theorem}{Theorem}
\newtheorem{definition}{Definition}
\newtheorem{lemma}{Lemma}
\newtheorem{remark}{Remark}
\newtheorem{observation}{Observation}

\newtheorem{proposition}{Proposition}

\usepackage{xcolor}

\newtheorem{innercustomthm}{Theorem}

\newenvironment{customthm}[1]
  {\renewcommand\theinnercustomthm{#1}\innercustomthm}
  {\endinnercustomthm}

\newcommand{\C}{\mathcal{C}}
\newcommand{\R}{\mathcal{R}}

\newcommand{\cycleconj}{\textsf{1-vs-2Cycle}}
\newcommand{\compress}{\blackbox{Compress}}
\newcommand{\rake}{\blackbox{Rake}}
\newcommand{\bigsmall}{$\alpha$-Big-Small tree}

\newcommand{\children}{\Gamma}
\newcommand{\leaves}{\mathcal{L}}
\newcommand{\degr}{deg}

\newcommand{\machines}{\mathcal{P}}
\newcommand{\memory}{\mathcal{S}}
\newcommand{\hashtable}{\mathcal{H}}

\newcommand{\degwei}{degree-weighted}

%% file: introduction.tex
\section{Introduction}


In this paper, we study and extend Miller and Reif's fundamental FOCS'85 ~\cite{miller1985parallel,miller1991parallel,miller1989parallel} $O(\log n)$-round parallel {\em tree contraction} method. Their work leverages PRAM, a model of computation in which a large number of processors operate synchronously under a single clock and are able to randomly access a large shared memory. In PRAM, tree contractions require $n$ processors. Though the initial study of tree contractions was in the CRCW (concurrent read from and write to shared memory) PRAM model, this was later extended to the stricter EREW (exclusive read from and write to shared memory) PRAM model~\cite{dekel1986parallel} as well, and then to work-optimal parallel algorithms with $O(n/\log n)$ processors~\cite{gazit1988optimal}. Since then, a number of additional works have also built on top of Miller and Reif's tree contraction algorithm~\cite{acar2004dynamizing,cole1988the,gibbons1989optimal}.  Tree-based computations have a breadth of applications, including natural graph problems like matching and bisection on trees, as well as problems that can be formulated on tree-like structures including expression simplification. 

The tree contraction method in particular is an extremely broad technique that can be applied to many problems on trees. Miller and Reif~\cite{miller1989parallel} initially motivated their work by showing it can be used to evaluate arithmetic expressions. They additionally studied a number of other applications~\cite{miller1991parallel}, using tree contractions to construct the first polylogarithmic round algorithm for tree isomorphism and maximal subtree isomorphism of unbounded degrees, compute the 3-connected components of a graph, find planar embeddings of graphs, and compute list-rankings. An incredible amount of research has been conducted to further extend the use of tree contractions for online evaluation of arithmetic circuits~\cite{miller1988efficient}, finding planar graph separators~\cite{gazit1987a}, approximating treewidth~\cite{bodlaender2016a}, and much more  \cite{atallah1989constructing,goodrich1996sorting,grohe2006testing,jez2016approximation,miller1987a,papadopoulos2015practical}. This work extends classic tree contractions to the adaptive massively parallel setting.

The importance of large-scale data processing has spurred a large interest in the study of massively parallel computing in recent years. Notably, the \textit{Massively Parallel Computation} (MPC) model has been studied extensively in the theory community for a range of applications \cite{ahn2015access,andoni2014parallel,andoni2018parallel,andoni2019log,assadi2019coresets,assadi2019sublinear,assadi2019massively,bateni2017affinity,bateni2018massively,behnezhad2019exponentially,behnezhad2019near,behnezhad2019streaming,boroujeni2018approximating,czumaj2018round,ghaffari2018improved,hajiaghayi2020matching,harvey2018greedy,lacki2020walking,nanongkai2019equivalence,roughgarden2016shuffles,yaroslavtsev2018massively}, many with a particular focus on graph problems. MPC is famous for being an abstraction of MapReduce~\cite{karloff2010a}, a popular and practical programming framework that has influenced other parallel frameworks including Spark~\cite{zaharia2016apache}, Hadoop~\cite{hadoop}, and Flume~\cite{chambers2010flumejava}. At a high level, in MPC, data is distributed across a range of low-memory machines which execute local computations in rounds. At the end of each round, machines are allowed to communicate using messages that do not exceed their local space constraints. In the most challenging space-constrained version of MPC, we restrict machines to $O(n^\epsilon)$ local space for a constant $0 < \epsilon < 1$ and $\widetilde{O}(n+m)$ total space (for graphs with $m$ edges, or just $\widetilde{O}(n)$ otherwise).

The computation bottleneck in practical implementations of massively parallel algorithms is often the amount of communication. Thus, work in MPC often focuses on \textit{round complexity}, or the number of rounds, which should be $O(\log n)$ at a baseline. More ambitious research often strives for sublogarithmic or even constant round complexity, though this often requires very careful methods.
Among others, a specific family of graph problems known as \textit{Locally Checkable Labeling} (LCL) problems -- which includes vertex coloring, edge coloring, maximal independent set, and maximal matching to name a few -- admit highly efficient MPC algorithms, and have been heavily studied during recent years~\cite{behnezhad2019exponentially,assadi2019sublinear,assadi2019coresets,behnezhad2019bmassively,ghaffari2020improved,ghaffari2019sparsifying,czumaj2018round}. Another consists of DP problems on sequences including edit distance~\cite{boroujeni2018approximating} and longest common subsequence~\cite{hajiaghayi2019massivelyb}, as well as pattern matching~\cite{hajiaghayi2019massivelya}. The round complexity of aforementioned MPC algorithms can be interpreted as the parallelization limit of the corresponding problems.

While MPC is generally an extremely efficient model, it is theoretically limited by the widely believed \cycleconj{} conjecture~\cite{ghaffari2019conditional}, which poses that distinguishing between a graph that is a single $n$-cycle and a graph that is two $n/2$-cycles requires $\Omega(\log n)$ rounds in MPC. This has been shown to imply lower bounds on MPC round complexity for a number of other problems, including connectivity~\cite{behnezhad2019near}, matching~\cite{ghaffari2019conditional,nanongkai2019equivalence}, clustering~\cite{yaroslavtsev2018massively}, and more \cite{andoni2019log,ghaffari2019conditional,lacki2020walking}. To combat these conjectured bounds, Behnezhad et al. \cite{behnezhad2019massively} developed a stronger and practically-motivated extension of MPC, called \textit{Adaptive Massively Parallel Computing} (AMPC). AMPC was inspired by two results showing that adding distributed hash tables to the MPC model yields more efficient algorithms for finding connected components~\cite{kiveris2014connected} and creating hierarchical clusterings~\cite{bateni2017affinity}. AMPC models exactly this: it builds on top of MPC by allowing in-round access to a distributed read-only hash table of size $O(n+m)$.
See Section~\ref{sec:ampc} for a formal definition.

In their foundational work, Behnezhad et al. \cite{behnezhad2019massively} design AMPC algorithms that outperform the MPC state-of-the-art on a number of problems. This includes solving minimum spanning tree and 2-edge connectivity in $\log\log_{m/n}(n)$ AMPC rounds (outperforming $O(\log n)$ and $O(\log D\log\log_{m/n} n)$ MPC rounds respectively), and solving maximal independent set, 2-\textsf{Cycle}, and forest connectivity in $O(1)$ AMPC rounds (outperforming $\widetilde{O}(\sqrt{\log n})$, $O(\log n)$, and $O(\log D\log\log_{m/n}n)$ MPC rounds respectively). Perhaps most notably, however, they proved that the \cycleconj{} conjecture does not apply to AMPC by finding an algorithm to solve connectivity in $O(\log\log_{m/n}n)$ rounds. This was later improved to be $O(1/\epsilon)$ by Behnezhad et al.~\cite{behnezhad2020parallel}, who additionally found improved algorithms for AMPC minimum spanning forest and maximum matching. Charikar,
 Ma, and Tan~\cite{charikar2020unconditional} very recently show that connectivity in the AMPC model requires $\Omega(1/\epsilon)$ rounds unconditionally, and thus the connectivity result of Behnezhad et al.~\cite{behnezhad2020parallel} is indeed tight. 
 

A notable drawback of the current work in AMPC is that there is no generalized framework for solving multiple problems of a certain class. Such methods are important for providing a deeper understanding of how the strength of AMPC can be leveraged to beat MPC in general problems, and often leads to solutions for entirely different problems. Studying Miller and Reif~\cite{miller1989parallel}'s tree contraction algorithm in the context of AMPC provides exactly this benefit. We get a generalized technique for solving problems on trees, which can be extended to a range of applications.

Recently, Bateni et al.~\cite{bateni2018mapreduce} introduced a generalized method for solving ``polylog-expressible'' and ``linear-expressible'' dynamic programs on trees  in the MPC model. This was heavily inspired by tree contractions, and also is a significant inspiration to our work. Specifically, their method solves minimum bisection, minimum $k$-spanning tree, maximum weighted matching, and a large number of other problems in $O(\log n)$ rounds. We  extend these methods, as well as the original tree contraction methods, to the AMPC model to create more general techniques that solve many problems in $O_\epsilon(1)$ rounds.

\input{ampc}

\input{intro-contributions}

\subsection{Paper Outline}
The work presented in this paper is a constant-round generalized technique for solving a large number of graph theoretic problems on trees in the AMPC model. In Section~\ref{sec:preliminaries}, we go over some notable definitions and conventions we will be using throughout the paper. This includes the introduction of a generalized weighted tree, a formalization of the general tree contraction process, the definition of contracting functions, and a discussion of a tree decomposition method we call the \emph{preorder decomposition}. In the Section~\ref{sec:contributions}, we go over our main results, algorithms, and proofs. The first result (\S\ref{sec:boundeddegree}) is an algorithm for executing a tree contraction-like process which solves the same problems on trees of bounded maximum degree. The second result (\S\ref{sec:main}) utilizes the first result as well as additional novel techniques to implement generalized tree contractions. We additionally show (\S\ref{sec:2contract}) that our algorithms can also implement Miller and Reif's standard notion of tree contractions, and (\S\ref{sec:reconstruct}) we show how to efficiently reconstruct a solution on the entire graph by reversing the tree contracting process. Finally, in Section~\ref{sec:applications}, we apply our algorithms to solve a number of popular problems on trees.

%% file: ampc.tex
\subsection{The AMPC Model}\label{sec:ampc}
The AMPC model, introduced by Behnezhad et. al~\cite{behnezhad2019massively}, is an extension of the standard MPC model with additional access to a \textit{distributed hash table}. In MPC, data is initially distributed across machines and then computation proceeds in rounds where machines execute local computations and then are able to share small messages with each other before the next round of computation. A distributed hash table stores a collection of key-value pairs which are accessible from every machine, and it is required that both key and value have a constant size. Each machine can adaptively query a bounded sequence of keys from a centralized distributed hash table during each round, and write a bounded number of key-value pairs to a distinct distributed hash table which is accessible to all machines in the next round. The distributed hash tables can also be utilized as the means of communication between the machines, which is implicitly handled in the MPC model, as well as a place to store the initial input of the problem. It is straight-forward to see how every MPC algorithm can be implemented within the same guarantees for the round-complexity and memory requirements in the AMPC model.


\begin{definition}
Consider a given graph on $n$ vertices and $m$ edges. In the \textbf{AMPC model}, there are $\machines$ machines each with \textbf{sublinear} local space $\memory  = O(n^\epsilon)$ for some constant $0 < \epsilon < 1$, and the total memory of machines is bounded by $\widetilde{O}(n + m)$. In addition, there exist a collection of distributed hash tables $\hashtable_0,\hashtable_1,\hashtable_2,\ldots$, where $\hashtable_0$ contains the initial input.

The process consists of several rounds. During round $i$, each machine is allowed to make at most $O(\memory)$ read queries from $\hashtable_{i-1}$ and to write at most $O(\memory)$ key-value pairs to $\hashtable_i$. Meanwhile, the machines are allowed to perform an arbitrary amount of computation locally. Therefore, it is 
possible for machines to decide what to query next after observing the result of previous queries. In this sense, the queries in this model are \textbf{adaptive}.

\end{definition}


%% file: intro-contributions.tex
\subsection{Our Contributions}

The goal of this paper is to present a framework for solving various problems on trees with constant-round algorithms in AMPC.
This is a general strategy, where we intelligently shrink the tree iteratively via a \emph{decomposition} and \emph{contraction} process.  Specifically, we follow Miller and Reif's~\cite{miller1989parallel} two-stage process, where we first \emph{compress} each connected component in our decomposition\footnote{Each group in our decomposition may consists of multiple connected components on the tree.}, and then \emph{rake} the leaves by contracting all leaves of the same parent together. We repeat until we are left with a single vertex, from which we can extract a solution. To retrieve the solution when the output corresponds to many vertices in the tree (i.e., maximum matching istead of maximum matching value), we can undo the contractions in reverse order and populate the output as we gradually reconstruct the original tree.

The decomposition strategy must be constructed very carefully such that we do not lose too much information to solve the original problem and each connected component must fit on a single machine with $O(n^\epsilon)$ local memory. 
To compress, we require oracle access to a black-box function, a \emph{connected contracting function}, which can efficiently contract a connected component into a vertex while also retaining enough information to solve the original problem. To rake leaves, we require oracle access to another block-box function, a \emph{sibling contracting function}, which executes the same thing but on a set of leaves that share a parent. These two black-box functions are problem specific (e.g., we need a different set of functions for maximum matching and maximum independent set). In this paper, 
we only require contracting functions to accept $n^\epsilon$ vertices as the input subgraphs, and we always run these black-box functions locally on a single machine. Thus, we can compress any arbitrary collection of disjoint components of size at most $n^\epsilon$ in $O(1)$ AMPC rounds. See Section \ref{subsec:contracting} for formal definitions.

This general strategy actually works on a special class of structures, called \emph{\degwei{} trees} (defined in \S\ref{sec:preliminaries}). Effectively, these are trees $T = (V, E, W)$ with a multi-dimensional weight function where $W(v) \in\{0,1\}^{\widetilde{O}(\deg(v))}$ stores a vector of bits proportional in size to the degree of the vertex $v \in V$. When we use our contracting functions, we use $W$ to store data about the set of vertices we are contracting. This is what allows our algorithms to retain enough information to construct a solution to the entire tree $T$ when we contract sets of vertices. Note that the degree of the surviving vertex after contraction could be much smaller than the total degree of the original set of vertices. 

Our first algorithm works on trees with bounded degree, more precisely, trees with maximum degree at most $n^\epsilon$. The reason this is easier is because when an internal connected component is contracted, we often need to encode the output of the subproblem at the root (e.g., the maximum weighted matching on the rooted subtree) in terms of the children of this component post-contraction. In high degree graphs, 
it may have many children after being contracted, and therefore require a large encoding (i.e., one larger than $O(n^\epsilon)$) and thus not fit on one machine. 

In this algorithm, we find that if the degree is bounded by $n^\epsilon$ and we compress sufficiently small components, then the algorithm works out much more smoothly. The underlying technique that allows us to contract the tree into a single vertex in $O(1/\epsilon)$ iterations is a decomposition of vertices based on their preorder numbering. The surprising fact is that each group in this decomposition contains at most one non-leaf vertex after contracting connected components. Thus, an additional single rake stage is sufficient to collapse any tree with $n$ vertices to a tree with at most $n^{1-\epsilon}$ vertices in a single iteration. However, we need $O(1/\epsilon)$ AMPC rounds at the beginning of each iteration to find the decomposition associated with the resulting tree after contractions performed in the previous iteration. This becomes $O(1/\epsilon^2)$ AMPC rounds across all iterations. See Section \ref{sec:boundeddegree} for the proofs and more details.

This is a nice independent result, proving a slightly more efficient $O(1/\epsilon^2)$-round algorithm on degree bounded trees. Additionally, many problems on larger degree trees can be represented by lower degree graphs. For example, both the original Miller and Reif~\cite{miller1985parallel} tree contraction and the Betani et al.~\cite{bateni2018mapreduce} framework consider only problems in which we can replace each high degree vertex by a balanced binary tree, reducing the tree-based computation on general trees to a slightly different computation on binary trees.
Equally notably, it is an important subroutine in our main algorithm.

\begin{restatable}{theorem}{boundedmain}
\label{thm:boundedmain}
Consider a \degwei{} tree $T = (V, E, W)$ and a problem $P$. Given a connected contracting function on $T$ with respect to $P$, one can compute $P(T)$ in $O(1/\epsilon^2)$ AMPC rounds with $O(n^\epsilon)$ memory per machine and $\widetilde{O}(n)$ total memory if $\deg(v)\leq n^\epsilon$ for every vertex $v\in V$.
\end{restatable}

\begin{remark}\label{remark:median} It may be tempting to suggest that in most natural problems the input tree can be transformed into a tree with degree bounded by $n^\epsilon$. However, we briefly pose the \textsf{MedianParent} problem, where leaves are given values and parents are defined recursively as the median of their children. By transforming the tree to make it degree bounded, we lose necessary information to find the median value among the children of a high degree vertex. 
\end{remark}

Next, we move onto our main result: a generalized tree contraction algorithm that works on any input tree with arbitrary structure. Building on top of Theorem~\ref{thm:boundedmain}, we can create a natural extension of tree contractions. Recall that the black-box contracting functions encode the data associated with a contracted vertex in terms of its children post-contraction. Thus, allowing high degree vertices introduces difficulties working with contracting functions. In particular, it is not possible to store the weight vector $W(v)$ of a high degree vertex $v$ inside the local memory of a single machine. The power of this algorithm is its ability to implement \compress{} and \rake{} for $n^\epsilon$-tree-contractions in $O(1/\epsilon^3)$ rounds. 

The most significant novelty of our main algorithm is the handling of high degree vertices. To do this, we first handle all maximal connected components of low degree vertices using the algorithm from Theorem~\ref{thm:boundedmain} as a black-box. This compresses each such component into one vertex without needing to handle high degree vertices. By contracting these components, we obtain a special tree called \emph{Big-Small-tree} (defined formally in \S\ref{sec:main}) which exhibits nice structural properties. Since the low degree components are maximal, the degree of each vertex in every other layer is at least $n^\epsilon$, implying an $O(1/\epsilon)$ upper-bound on the depth of Big-Small-trees. Hence, after a single rake stage, the number of high degree vertices drops by a factor of $n^\epsilon$.

In order to rake the leaves of high degree vertices, we have to carefully apply our sibling contracting functions in a way that can be implemented efficiently in AMPC. Unlike Theorem~\ref{thm:boundedmain} in which having access to a connected contracting function is sufficient, here we also require a sibling contracting function. Consider a star tree with its center at the root. Without a sibling contracting function, we are able to contract at most $O(n^\epsilon)$ vertices in each round since the components we pass to the contracting functions must be disjoint. But having access to a sibling contracting function, we can rake up to $O(n)$ leaf children of a high degree vertex in $O(1/\epsilon)$ rounds. For more details about the algorithm and proofs see Section~\ref{sec:main}.



\begin{restatable}{theorem}{main}
\label{thm:main}
Consider a \degwei{} tree $T = (V, E, W)$ and a problem $P$. Given a connected contracting function and a sibling contracting function on $T$ with respect to $P$, one can compute $P(T)$ in $O(1/\epsilon^3)$ AMPC rounds with $O(n^\epsilon)$ memory per machine and $\widetilde{O}(n)$ total memory.
\end{restatable}

Theorem~\ref{thm:boundedmain} and Theorem~\ref{thm:main} give us general tools that have the power to create efficient AMPC algorithms for any problem that admits a connected contracting function and a sibling contracting function. Intuitively, they reduce constant-round parallel algorithms for a specific problem on trees to designing black-box contracting functions that are sequential.
We should be careful in designing contracting functions to make sure that the amount of data stored in the surviving vertex does not asymptotically exceed its degree in the contracted tree. Also note that a connected contracting function works with unknown values that depend on the result of other components. 

Satisfying these conditions is a factor that limits the extent of problems that can be solved using our framework. For example, the framework of Bateni et. al~\cite{bateni2018mapreduce} works on a wider range of problems on trees since their algorithm, roughly speaking, tolerates exponential growth of weight vectors using a careful decomposition of tree. Indeed, they achieve these benefits at the cost of an inherent requirement for at least $O(\log n)$ rounds due to the divide-and-conquer nature of their algorithm. However, their framework comes short on addressing problems such as \textsf{MedianParent} (defined in Remark~\ref{remark:median}) that are not reducible to binary trees.
Nonetheless, we show several techniques for designing contracting functions that satisfy these conditions, in particular:

\begin{enumerate}
    \item In Section~\ref{sec:2contract}, we prove a general approach for designing a connected contracting function and a sibling contracting function given a PRAM algorithm based on the original Miller and Reif~\cite{miller1985parallel} tree contraction. We do this by observing that in almost every conventional application of Miller and Reif's framework, the length of data stored at each vertex remains constant throughout the algorithm.
    \item Storing a minimal tree representation of a connected component contracted into $v$ in the weight vector $W(v)$ enables us to simplify a recursive function defined on the subtree rooted at $v$ in terms of yet-unknown values of its children, while keeping the length of $W(v)$ asymptotically proportional to $deg(v)$. For instance, see Section~\ref{sec:maximummatching} which utilizes this approach in the context of maximum weighted matching.
\end{enumerate}

Ultimately, this is a highly efficient generalization of the powerful tree contraction algorithm. To illustrate the versatility of our framework, we show that it gives us efficient AMPC algorithms for many important applications of frameworks such as Miller and Reif~\cite{miller1989parallel}'s and Bateni et al.~\cite{bateni2018mapreduce}'s by constructing sequential black-box contracting functions. In doing so, we utilize a diverse set of techniques, including the ones mentioned above, that are of independent interest and can be applied to a broad range of problems on trees.



\begin{restatable}{theorem}{applications}
\label{thm:applications}
Algorithms~\ref{alg:boundedtreecontract} and~\ref{alg:main} can solve, among other applications, dynamic expression evaluation, tree isomorphism testing, maximal matching, and maximal independent set in $O(1/\epsilon^2)$ AMPC rounds, and maximum weighted matching and maximum weighted independent set in $O(1/\epsilon^3)$ AMPC rounds. All algorithms use $O(n^\epsilon)$ memory per machine and $\widetilde{O}(n)$ total memory.
\end{restatable}

%% file: preliminaries.tex
\section{Preliminaries}\label{sec:preliminaries}
In this work, we are interested in solving problems on trees $T=(V,E)$ where $|V| = n$. Our algorithms iteratively transform $T$ by contracting components in an intelligent way that: (1) components can be stored on a single machine, (2) the number of iterations required to contract $T$ to a single vertex is small, and (3) at each step of the process, we still have enough information to solve the initial problem on $T$.

To achieve (3), we must retain some information about an original component after we contract it. For instance, consider computing all maximal subtree sizes. For a connected component $S$ with $r= \text{lca}(S)$\footnote{lca is the least common ancestor function.}, the contracted vertex $v_S$ of $S$ might encode $|S|$ and a list of its leaves (when viewing $S$ as a tree itself). It is not difficult to see that this would be sufficient knowledge to compute all maximal subtree sizes for the rest of the vertices in $T$ without considering all individual vertices in $S$. Data such as this is encoded as a multi-dimensional weight function which maps vertices to binary vectors. We will specifically consider trees where the dimensionality of the weight function is bounded by the degree of the vertex.

We note that in this paper, when we refer to the degree of a vertex in a rooted tree, we ignore parents. Therefore, $\deg(v)$ is the number of children a vertex has.

\begin{definition}
A \textbf{\degwei{} tree} is a tree $T = (V, E, W)$ with vertex set $V$, edge set $E$, and vertex weight vector function $W$ such that for all $v\in V$, $W(v) \in \{0,1\}^{\widetilde{O}(\deg(v))}$.\footnote{$\widetilde{O}(f(n)) = O(f(n)\log n)$.}
\end{definition}

Notationally, we let $w(v) = \dim(W(v)) = \widetilde{O}(\deg(v))$ be the length of the weight vectors. Additionally, note that a tree $T=(V,E)$ is a \degwei{} tree where $W(v) = \emptyset$ for all $v\in V$.

In order to implement our algorithm, we also require specific \emph{contracting functions} whose properties allow us to achieve the desired result (\S\ref{subsec:contracting}). In addition, we will introduce a specific tree decomposition method, called a \emph{preorder decomposition}, that we will efficiently implement and leverage in our final algorithms (\S\ref{subsec:decompose}). 

\subsection{Tree Contractions and Contracting Functions}\label{subsec:contracting}

Our algorithms provide highly efficient generalizations to Miller and Reif's~\cite{miller1989parallel} tree contraction algorithms. At a high level, their framework provides the means to compute a global property with respect to a given tree in $O(\log{n})$ \emph{phases}. In each \emph{phase}, there are two \emph{stages}:

\begin{itemize}
    \item \compress{} stage: Contract around half of the vertices with degree $1$ into their parent.
    \item \rake{} stage: Contract all the leaves (vertices with degree $0$) into their parent.
\end{itemize}

Repeated application of \compress{} and \rake{} alternatively results in a tree which has only one vertex. 
Intuitively, the \compress{} stage aims to shorten the long \emph{chains}, maximal connected sequences of vertices whose degree is equal to $1$, and the \rake{} stage cleans up the leaves.
Both stages are necessary in order to guarantee that $O(\log{n})$ phases are enough to end up with a single remaining vertex~\cite{miller1989parallel}. 

In the original variant, every odd-indexed vertex of each \emph{chain}
is contracted in a \compress{} stage. In some randomized variants, each vertex is selected with probability $1/2$ independently, and an independent set of the selected vertices is contracted. In such variants, contracting two consecutive vertices in a chain is avoided in order to efficiently implement the tree contraction in the PRAM model. However, this restriction is not imposed in the AMPC model, and hence we consider a more \emph{relaxed variant} of the \compress{} stage where each maximal chain is contracted into a single vertex.

We introduce a more generalized version of tree contraction called \emph{$\alpha$-tree-contractions}. Here, the \rake{} stage is the same as before, but in the \compress{} stage, every maximal subgraph containing only vertices with degree less than $\alpha$ is contracted into a single vertex.

\begin{definition}\label{def:alpha-tc}
In an $\alpha$-tree-contraction of a tree $T=(V,E)$, we repeat two stages in a number of phases until the whole tree is contracted into a single vertex:
\begin{itemize}
    \item \blackbox{Compress} stage: Contract every maximal connected component $S$ containing only vertices with degree less than $\alpha$, i.e., $\degr(v)<\alpha\;\;\forall{v \in S}$, into a single vertex $S'$.
    \item \blackbox{Rake} stage: Contract all the leaves into their parent.
\end{itemize}
\end{definition}

Notice that the relaxed variant of Miller and Reif's \compress{} stage is the special case when $\alpha=2$. Our goal will be to implement efficient $\alpha$-tree-contractions where $\alpha=n^\epsilon$.


In order to implement \compress{} and \rake{}, we need fundamental tools for contracting a single set of vertices into each other. We call these \textit{contracting functions}. In the \compress{} stage, we must contract connected components. In the \rake{} stage, we must contract leaves with the same parent into a single vertex. These functions run locally on small sets of vertices.

\begin{definition}\label{def:contractingfunction}
Let $P$ be some problem on \degwei{} trees such that for some \degwei{} tree $T$, $P(T)$ is the solution to the problem on $T$. A \textbf{contracting function} on $T$ with respect to $P$ is a function $f$ that replaces a set of vertices in $T$ with a single vertex and incident edges to form a \degwei{} tree $T'$ such that $P(T) = P(T')$\footnote{With some nuance, it depends on the format of the problem. For instance, when computing the value of the maximum independent set, the single values $P(T)$ and $P(T')$ should be the same. When computing the maximum independent set itself, uncontracted vertices must have the same membership in the set, and contracted vertices represent their roots.}. There are two types:
\begin{enumerate}
\item $f$ is a \textbf{connected contracting function} if $f$ contracts\footnote{Consider a connected component $S$ with a set of external neighbors $N(S) = \{v\in V\setminus S:\exists u\in S (v,u) \in E\}$. Then contracting $S$ means replacing $S$ with a single vertex with neighborhood $N(S)$.} connected components into a single vertex of $T$.
\item $f$ is a \textbf{sibling contracting function} if $f$ is defined on sets of leaf siblings (i.e., leaves that share a parent $p$) of $T$, and the new vertex is a leaf child of $p$.
\end{enumerate}
\end{definition}

Since the output of the contracting function is a \degwei{} tree, it implicitly must create a weight $W(v)$ for any newly contracted vertex $v$.


\subsection{Preorder Decomposition}\label{subsec:decompose}

A \textit{preorder decomposition} (formally defined shortly) is a strategy for decomposing trees into a disjoint union of (possibly not connected) vertex groups. In this paper, we will show that the preorder decomposition exhibits a number of nice properties (see \S\ref{sec:contributions}) that will be necessary for our tree contraction algorithms. Ultimately, we wish to find a decomposition of vertices $V_1, V_2, \ldots, V_k \subseteq V$ of a given tree $T=(V,E)$ ($\cup_{i=1}^k{V_i} = V$ and $V_i \cap V_j = \phi \;\;\forall{i,j}: i \neq j$) so that for all ${i \in [k]}$, after contracting each connected component contained in the same vertex group, the maximum degree is bounded by some given $\lambda$. Obviously, this won't be generally possible (i.e., consider a large star), but we will show that this holds when the maximum degree of the input tree is bounded as well.

The preorder decomposition is depicted in Figure~\ref{fig:preorder}. Number the vertices by their index in the preorder traversal of tree $T$, i.e., vertices are numbered $1,2,\ldots, n$ where vertex $i$ is the $i$-th vertex that is visited in the preorder traversal starting from vertex $1$ as root. In a preorder decomposition of $T$, each group $V_i$ consists of a consecutive set of vertices in the preorder numbering of the vertices. More precisely, let $l_i$ denote the index of the vertex $v \in V_i$ with the largest index, and assume $l_0 = 0$ for consistency. In a preorder decomposition, group $V_i$ consists of vertices $l_{i-1} + 1, l_{i-1} + 2, \ldots, l_{i}$.

\begin{definition}\label{def:preorder-dependency}
Given a tree $T=(V,E)$, a \emph{``preorder decomposition''} $V_1,V_2,\ldots,V_k$ of $T$ is defined by a vector $l \in \mathbb{Z}^{k+1}$, such that $0 = l_0 < l_1 < \ldots < l_k = n$, as $V_i = \{l_{i-1}+1,l_{i-1}+2,\ldots,l_i\}\;\; \forall{i \in [k]}$. See Subfigure~\ref{fig:preorder} for an example.
\end{definition}

Assume we want each $V_i$ in our preorder decomposition to satisfy $\sum_{v\in V_i} \deg(v) \leq \lambda$ for some $\lambda$. As long as $\deg(v) \leq \lambda$ for all $v\in V$, we can greedily construct components $V_1,\ldots,V_k$ according to the preorder traversal, only stopping when the next vertex violates the constraint. Since $\sum_{v\in V} \deg(v) \leq n$, it is not hard to see that this will result in $O(n/\lambda)$ groups that satisfy the degree sum constraint.

\begin{observation}\label{obs:groups}
Consider a given tree $T = (V, E)$. For any parameter $\lambda$ such that $\deg(v) \leq \lambda$ for all $v \in V$, there is a preorder decomposition $V_1, V_2, \ldots, V_k$ such that $\forall{i \in [k]}$, $\sum_{v\in V_i} \deg(v)\leq \lambda$, and $k = O(n/\lambda)$.
\end{observation}

The \emph{dependency tree} $T'=(V', E')$, as seen in Figure~\ref{fig:dependency} of a decomposition is useful notion for understanding the structure of the resulting graph. In $T'$, vertices represent connected components within groups, and there is an edge between vertices if one contains a vertex that is a parent of a vertex in the other. This represents our contraction process and will be useful for bounding the size of the graph after each step.

\begin{definition}\label{def:dependency-tree}
Given a tree $T=(V, E)$ and a decomposition of vertices $V_1,V_2,\ldots,V_k$, the \textbf{dependency tree} $T'=(V',E')$ of $T$ under this decomposition is constructed by contracting each connected component $C_{i,j}$ for all $j \in [c_i]$ in each group $V_i$. We call a component contracted to a leaf in $T'$ an \textbf{independent} component, and a component contracted to a non-leaf vertex in $T'$ a \textbf{dependent} component.
\end{definition}

\begin{figure}[ht]
    \centering
    \begin{subfigure}[b]{0.50\textwidth}
        \centering
        \includegraphics[scale=0.05]{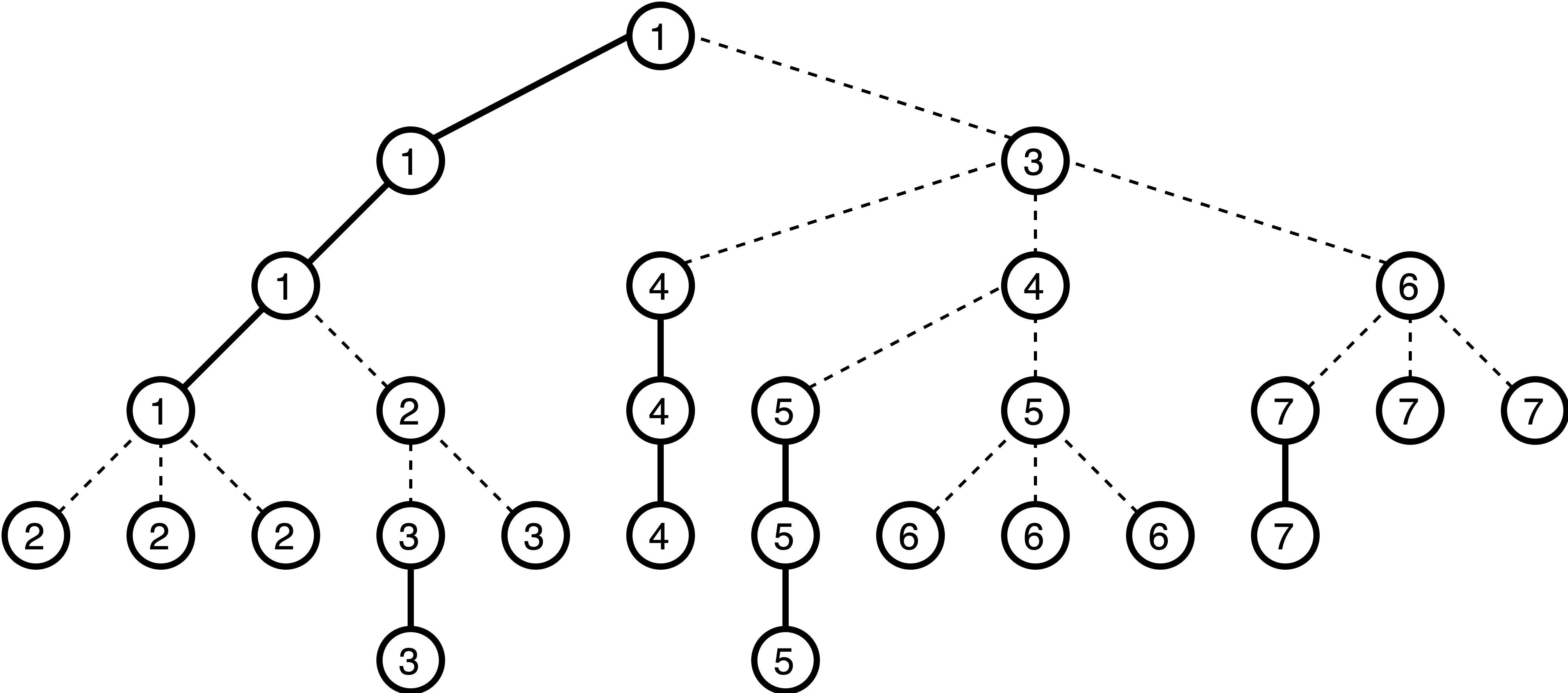}
        \caption{An example preorder decomposition of $T$ into $V_1,V_2,\ldots,V_7$ with  $\lambda = 8$. Edges within any $F_i$ are depicted bold, and edges belonging to no $F_i$ are depicted dashed.}
        \label{fig:preorder}
    \end{subfigure}\qquad
    \begin{subfigure}[b]{0.40\textwidth}
        \centering
        \includegraphics[scale=0.065]{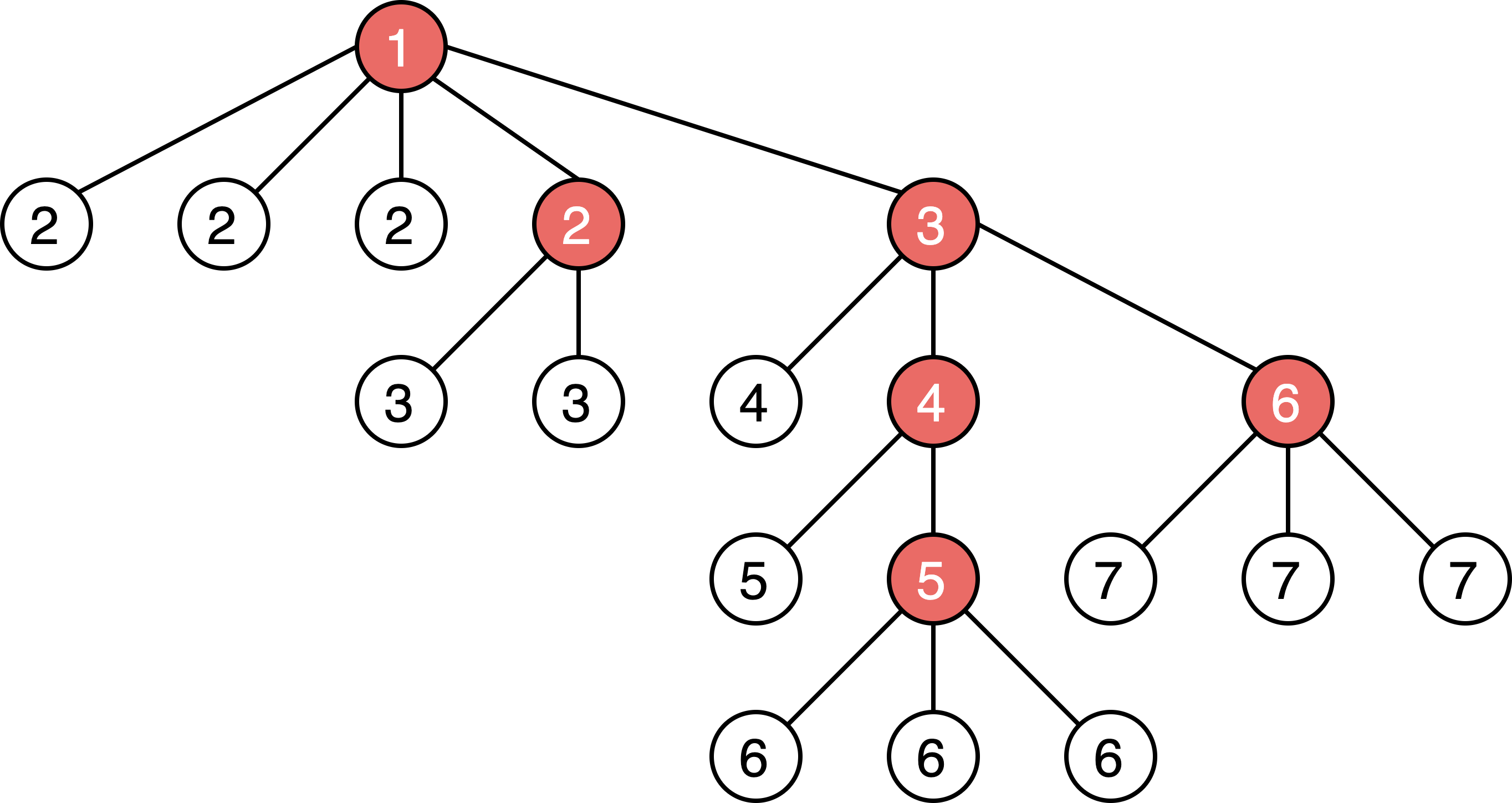}
        \caption{Dependency tree $T'$, created by contracting connected components of every $F_i$. Each red vertex represents a dependent component, and each white vertex represents an independent component.}
        \label{fig:dependency}
    \end{subfigure}
    \caption{In Subfigure (\subref{fig:preorder}), a preorder decomposition of a given tree $T$ is demonstrated. Based on this preorder decomposition, we define a dependency tree $T'$ so that each connected component $S$ in each forest $F_i$ is contracted into a single vertex $S'$. This dependency tree $T'$ is demonstrated in Subfigure (\subref{fig:dependency}). It is easy to observe that the contracted components are maximal components which are connected using bold edges in $T$, and each edge in $T'$ corresponds to a dashed edge in $T$.} 
    \label{fig:preorderdecomp}
\end{figure}

%% file: contributions.tex
\section{Constant-round Tree Contractions in AMPC}\label{sec:contributions}

The main results of this paper are two new algorithms. The first algorithm applies $\alpha$-tree-contraction-\emph{like} methods in order to solve problems on trees where the degrees are bounded by $n^\epsilon$. Though this algorithm is similar in inspiration to the notion of tree contractions, it is not a true $\alpha$-tree-contraction method.

\boundedmain*

This algorithm provides us with two benefits: (1) it is a standalone result that is quite powerful in its own right and (2) it is leveraged in our main algorithm for Theorem~\ref{thm:main}. The only differences between this result and our main result for generalized tree contractions is that we require $\deg(v)\leq n^\epsilon$, but it runs in $O(1/\epsilon^2)$ rounds, as opposed to $O(1/\epsilon^3)$ rounds. Thus, if the input tree has degree bounded by $n^\epsilon$, then clearly the precondition is satisfied. Additionally, if the tree can be decomposed into a tree with bounded degree such that we can still solve the problem on the decomposed tree, this result applies as well.

Our general results are quite similar, with a slightly worse round complexity, but with the ability to solve the problem on all trees. Notably, it is a true $\alpha$-tree-contraction algorithm.

\main*

In this section, we introduce both algorithms and prove both theorems.

\subsection{Contractions on Degree-Bounded Trees}\label{sec:boundeddegree}
\input{boundeddegree}



\subsection{Generalized $\alpha$-Tree-Contractions}
\label{sec:main}


In the rest of this section we prove our main result: a generalized tree contraction algorithm, Algorithm~\ref{alg:main}. Building on top of Theorem~\ref{thm:boundedmain}, we can create a natural extension of tree contractions. Recall from \S\ref{sec:preliminaries} that in the \compress{} stage, we must contract maximal connected components containing only vertices $v$ with degree $d(v) < \alpha$. Conveniently, by  Theorem~\ref{thm:boundedmain}, Algorithm~\ref{alg:boundedtreecontract} achieves precisely this. Therefore, to implement tree contractions, we simply need to:

\begin{enumerate}
\item Identify maximal connected components of low degree (Algorithm~\ref{alg:main}, line \ref{line:connectivity}), which can be done in $O(1/\epsilon)$ rounds by Behnezhad et al.~\cite{behnezhad2020parallel}.
\item Use our previous algorithm to execute the \compress{} stage on each component (Algorithm~\ref{alg:main}, line \ref{line:boundedtreecontract}), which can be done by Algorithm~\ref{alg:boundedtreecontract} in $O(1/\epsilon^2)$ rounds.
\item Apply a function that can execute the \rake{} stage (Algorithm~\ref{alg:main}, lines \ref{line:rake-start} through \ref{line:rake-end}).
\end{enumerate}

To satisfy the third step, we use  a \textit{sibling contracting function} (Definition~\ref{def:contractingfunction}), which can contract leaf-siblings of the same parent into a single leaf. Since a vertex might have up to $n$ children, to do this in parallel, we may have to group siblings into $n^\epsilon$-sized groups and repeatedly contract until we reach one leaf. Assuming sibling contractions are locally performed inside machines, this will then take $O(1/\epsilon)$ AMPC rounds.

\begin{algorithm}[ht]
\SetAlgoLined
\KwData{Degree-weighted tree $T = (V, E, W)$, a connected contracting function $\C$, and a sibling contracting function $\R$.}
\KwResult{The problem output $P(T)$.}
 $T_0 \leftarrow T$\;
 \For{$i\leftarrow 1$ \KwTo $l = O(1/\epsilon)$}{ 
    Let $\mathcal{S}_{i-1} \gets \textsf{Connectivity}(T_{i-1}\setminus\{v\in V:\degr(v)>n^\epsilon)$\;\label{line:connectivity}
    Let $\mathcal{K}_{i-1} \gets$ Components in $\mathcal{S}_{i-1,j}$ which represent a leaf in $T'_{i-1}$\;
    Contract each $S_{i-1,j} \in \mathcal{K}_{i-1}$ into $S'_{i-1,j}$ by applying $\textsf{BoundedTreeContract}(S_j, \C )$\;\label{line:boundedtreecontract}
    Let $\mathcal{L}_{i-1}$ be the set of all maximal leaf-stars (excluding their parent) in $T'_{i-1}$\;
    \For{$L_{i-1,0} = \{ v_1, v_2, \ldots, v_k \}\in \mathcal{L}_{i-1}$}{\label{line:rake-start}
        \For{$j\leftarrow 1$ \KwTo $1/\epsilon$}{
            Split $L_{i-1,j-1}$ into $k/n^{j\epsilon}$ parts $L_{i-1,j,1}, \ldots, L_{i-1, j,k/n^{j\epsilon}}$ each of size $n^\epsilon$\;
            Contract each $L_{i-1,j,z}$ into $L'_{i-1,j,z}$ by applying $\R(L_{i-1,j, z})$\;
            Let $L_{i-1,j} \gets \{ L'_{i-1,j,1}, \ldots, L'_{i-1,j,k/n^{j\epsilon}} \}$\;
        }
        Contract $L_{i-1,1/\epsilon}$ by applying $\R(L_{i-1,1/\epsilon})$\;
    }\label{line:rake-end}
    Let $T_i\gets T_{i-1}'$\;
  }
  \KwRet{$\C(T_l)$}\;
 \caption{\textsf{TreeContract} \\
 (Computing the solution $P(T)$ of a problem $P$ on degree-weighted tree $T$ using a connected contracting function $\mathcal{C}$ and a sibling contracting function $\mathcal{R}$)}
\label{alg:main}
\end{algorithm}

We can show that this requires $O(1/\epsilon)$ phases to execute, and  each phase takes $O(1/\epsilon^2)$ rounds to compute due to Theorem~\ref{thm:boundedmain} and our previous argument for \rake{} by sibling contraction. Thus we achieve the following result:

\main*

Recall the definition of $\alpha$-tree-contraction (Definition~\ref{def:alpha-tc}) from \S\ref{subsec:contracting}. First, we prove Lemma~\ref{lemma:rounds} to bound the number of phases in $\alpha$-tree-contraction.


\begin{restatable}{lemma}{alpharounds}\label{lemma:rounds}
For any $\alpha \geq 2$, the number of $\alpha$-tree-contraction phases until we have a constant number of vertices is bounded by $O(\log_{\alpha}(n))$.
\end{restatable}


To show Lemma~\ref{lemma:rounds} we will introduce a few definitions. The first definition we use is a useful way to represent the resulting tree after each \compress{} stage. Before stating the definition, recall the \emph{Dependency Tree} $T'$ of a tree $T$ from Definition~\ref{def:dependency-tree}.

\begin{definition}
An \textbf{$\alpha$-Big-Small Tree} $T'$ is the \emph{dependency tree} of a tree $T$ with weighted vertices if it is a minor of $T$ constructed by contracting all components of $T$ made up of low vertices $v$ with $deg(v) < \alpha$ (i.e., the connected components of $T$ if we were to simply remove all vertices $u$ with $deg(v) \geq \alpha$) into a single node.

We call a node $v$ in $T'$ with $deg(v) > \alpha$ in $T$ a \textbf{big node}. All other nodes in $T'$, which really represent contracted components of small vertices in $T$, are called \textbf{small components}. 
\end{definition}

Note that a small component may not be small in itself, but it can be broken down into smaller vertices in $T$. It is not hard to see the following simple property. This simply comes from the fact that maximal components of small degree vertices are compressed into a single small component, thus no two small components can be adjacent.

\begin{observation}\label{obs:bigsmall}
No small component in an \bigsmall{} can be the parent of another small component.
\end{observation}

Consider our dependency tree $T'$ based off a tree $T$ that has been compressed. Obviously, $T'$ is a minor of $T$ constructed as described for $\alpha$-Big-Small Tree because the weight of a vertex equals its number of children (by the assumption of Lemma~\ref{lemma:rounds}). Note a small component refers to the compressed components, and a big node refers to nodes that were left uncompressed. 

To show Lemma~\ref{lemma:rounds}, we start by proving that the ratio of leaves to nodes in $T'$ is large. Since \rake{} removes all of these leaves, this shows that $T$ gets significantly smaller at each step. Showing that the graph shrinks sufficiently at each phase will ultimately give us that the algorithm terminates in a small number of phases.

\begin{lemma}\label{lemma:reduceverts}
Let $T_i'$ be the tree at the end of phase $i$. Then the fraction of nodes that are leaves in $T_i'$ is at least $\alpha/(\alpha+4)$ as long as $w(v)$ is equal to the number of children of $v$ for all $v\in T_i'$ and $\alpha \geq 2$.
\end{lemma}

\begin{proof}
For our tree $T_i'$, we will call the number of nodes $n$, the number of leaves $\ell$, and the number of big nodes $b$. We want to show that $\ell > n\alpha/(\alpha+4)$. We induct on $b$. When $b=0$, we can have one small component in our tree, but no others can be added by Observation~\ref{obs:bigsmall}. Then $n=\ell=1$, so $\ell> n\alpha/(\alpha+4)$. 

Now consider $T_i'$ has some arbitrary $b$ number of big nodes. Since $T_i'$ is a tree, there must be some big node $v$ that has no big node descendants. Since all of its children must be small components and they cannot have big node descendants transitively, then Observation~\ref{obs:bigsmall} tells us each child of $v$ is a leaf. Note that since $v$ is a big node, it must have weight $w(v) > \alpha$, which also means it must have at least $\alpha$ children (who are all leaves) by the assumption that $w(v)$ is equal to the number of children.

Consider trimming $T_i'$ on the edge just above $v$. The size of this new graph is now $n^* = n - w(v) - 1$. It also has exactly one less big node than $T_i'$.  Therefore, inductively, we know the number of leaves in this new graph is at least $\ell^* \geq \frac{\alpha}{\alpha+4}n^* = \frac{\alpha}{\alpha+4}(n-w(v)-1)$. Compare this to the original tree $T_i'$. When we replace $v$ in the graph, we remove up to one leaf (the parent $p$ of $v$, if $p$ was a leaf when we cut $v$), but we add $w(v)$ new leaves. This means the number of leaves in $T_i'$ is:
\begin{align*}
\ell =& \ell^* - 1 + w(v)
\\=& \frac{\alpha}{\alpha+4}(n-w(v)-1) - 1 + w(v)
\\=& \frac{\alpha}{\alpha+4}n - \frac{\alpha}{\alpha+4}w(v) - \frac{\alpha}{\alpha+4} - 1 + w(v)
\\=& \frac{\alpha}{\alpha+4}n + \frac{4}{\alpha+4}w(v) - \frac{2\alpha+4}{\alpha+4}
\\>& \frac{\alpha}{\alpha+4}n + \frac{4}{\alpha+4}\alpha - \frac{2\alpha+4}{\alpha+4}\tag{1}\label{line:big}
\\\geq& \frac{\alpha}{\alpha+4}n\tag{2}\label{line:alpha}
\end{align*}

Where in line~(\ref{line:big}) we use that $w(v) > \alpha$ and in line~(\ref{line:alpha}) we use that $\alpha \geq 2$.
\end{proof}

Now we can prove our lemma.

\begin{proof}[Proof of Lemma~\ref{lemma:rounds}]
To show this, we will prove that the number of nodes from the start of one \compress{} to the next is reduced significantly. Consider $T_i$ as the tree before the $i$th \compress{} and $T_i'$ as the tree just after.  Let $T_{i+1}$ be the tree just before the $i+1$st \compress, and let $n_i$ be the number of nodes in $T_i$, $n_i'$ be the number of nodes in $T_i'$, and $n_{i+1}$ be the number of nodes in $T_{i+1}$. Since $T_i'$ is a minor of $T_i$, it must have at most the same number of vertices as $T_i$, so $n_i'\leq n_i$. Since $T_{i+1}$ is formed by applying \rake{} to $T_i'$, then it must have the number of nodes in $T_i'$ minus the number of leaves in $T_i'$ ($\ell_i'$). Therefore:

\[n_{i+1} = n_i' - \ell_i' \leq n_i' - \frac{\alpha}{\alpha+4}n_i' = \frac4{\alpha+4}n_i' \leq \frac4{\alpha+4}n_i\]

Where we apply both Lemma~\ref{lemma:reduceverts} that says $\ell_i' \geq \frac\alpha{\alpha+4}n_i'$ and the fact that we just showed that $n_i'\leq n_i$. This shows that from the start of one compress phase to another, the number of vertices reduces by a factor of $\frac4{\alpha+4}$. Therefore, to get to a constant number of vertices, we require $\log_{\frac{\alpha+4}{4}}(n) = O(\log_\alpha(n))$ phases.

\end{proof}

Now we are ready to prove our main theorem.

\main*

\begin{proof}
We will show that our Algorithm~\ref{alg:main}
achieves this result. Lemma~\ref{lemma:rounds} shows that there will be only at most $O(1/\epsilon)$ phases. In each phase $i$, we start by running a connectivity algorithm to find maximally connected components of bounded degree, which takes $O(1/\epsilon)$ time.  Let $\mathcal{K}_{i-1}$ be the set of connected components which are leaves in $T'_{i-1}$.  Then for each component $S_{i-1,j} \in \mathcal{K}_{i-1}$, we run \textsf{BoundedTreeContract}  (Algorithm~\ref{alg:boundedtreecontract}) in parallel
using only our connected contracting function $\C$.
Since the total degree of vertices over all members of $\mathcal{K}_{i-1}$ is not larger than $|T_{i-1}|$ and the amount of memory required for storing a degree-weighted trees is not larger than the total degree, the total number of machines is bounded above by $O(n^{1-\epsilon})$. By definition, the maximum degree of any $S_{i-1,j}$ is $n^\epsilon$. By Theorem~\ref{thm:boundedmain}, each instance of \textsf{BoundedTreeContract} requires $O(1/\epsilon^2)$ rounds, $O(|S_{i-1,j}|^\epsilon)$ local memory and $\widetilde{O}(|S_{i-1,j}|)$ total memory. As $|S_{i-1,j}| \leq |T_{i-1}|$ (we know $|T_0| = n$ and it only decreases over time), we only require at most $O(n^\epsilon)$ memory per machine. Since Since the total degree of vertices over all members of $\mathcal{K}_{i-1}$ is not larger than $|T_{i-1}|$, the total memory required is only $\widetilde{O}(|T_{i-1}|)) = \widetilde{O}(n)$. This is within the desired total memory constraints.

Finally, $\mathcal{R}$ is given to us as a sibling contractor. Consider the \rake{} stage in our algorithm. We distribute machines across maximal leaf-stars. For any leaf-star with $n^\epsilon \leq \degr(v) \leq kn^\epsilon$ for some (possibly not constant) $k$, we will allocate $k$ machines to that vertex. Since again the number of vertices is bounded above by $n$, this requires only $O(n^{1-\epsilon})$ machines. On each machine, we allocate up to $n^\epsilon$ leaf-children to contract into each other. We can then contract siblings into single vertices using $\mathcal{R}$. Since there are at most $n$ children for a single vertex, it takes at most $O(1/\epsilon)$ rounds to contract all siblings into each other. Then, finally, we can use $\mathcal{C}$ to compress the single child into its parent, which takes constant time.

Therefore, we have $O(1/\epsilon)$ phases which require $O(1/\epsilon^2)$ rounds each, so the total number of rounds is at most $O(1/\epsilon^3)$. We have also showed that throughout the algorithm, we maintain $O(n^\epsilon)$ memory per machine and $\widetilde{O}(n)$ total memory. This concludes the proof.
\end{proof}


\input{simulate}

%% file: boundeddegree.tex
We now provide an $O(1/\epsilon^2)$-round AMPC algorithm with local space $O(n^\epsilon)$  for solving any problem $P$ on a degree-weighted tree $T = (V, E, W)$ with bounded degree $\deg(v) \leq n^\epsilon$ for all $v\in V$ given a \emph{connected contracting function} for $P$. The method, which we call \textsf{BoundedTreeContract}, can be seen in Algorithm~\ref{alg:boundedtreecontract}.

Much like an $\alpha$-tree-contraction algorithm, it can be divided into a \compress{} and \rake{} stage. In the \compress{} stage, instead of compressing the whole maximal components that consist of low-degree vertices as required for $\alpha$-tree-contractions, we partition the vertices into groups using a preorder decomposition and bounding the group size by $n^\epsilon$. In the \rake{} stage, since the degree is bounded by $n^\epsilon$, all leaves who are children of the same vertex can fit on one machine. Thus each sibling contraction that must occur can be computed entirely locally. If we include the parent of the siblings, we can simply apply \compress{}'s connected contracting function on the children. This is why we do not need a sibling contracting function.

Let $T_0=T$ be the input tree. For every iteration $i\in[O(1/\epsilon)]$: (1) find a preorder decomposition $V_1,\ldots V_k$ of $T_{i-1}$, (2) contract each connected component in the preorder decomposition, and (3) put each maximal set of leaf-siblings (i.e., leaves that share a parent) in one machine and contract them into their parent. We sometimes refer to these maximal sets of leaf-siblings by \emph{leaf-stars}. After sufficiently many iterations, this should reduce the problem to a single vertex, and we can simply solve the problem on the vertex.

\begin{algorithm}[ht]
\SetAlgoLined
\KwData{Degree-weighted tree $T = (V, E, W)$ with degree bounded by $n^\epsilon$ and a connected contracting function $\mathcal{C}$.}
\KwResult{The problem output $P(T)$.}
 $T_0  \leftarrow T$\;
 \For{$i\leftarrow 1$ \KwTo $l = O(1/\epsilon)$}{
    Let $\mathcal{O}$ be a preordering of $V_{i-1}$\;
    Find a \emph{preorder decomposition} $V_1,V_2,\ldots,V_k$ of $T_{i-1}$ with $\lambda=n^\epsilon$ using $\mathcal{O}$\;
    Let $\mathcal{S}_{i-1}$ be the set of all connected components in $V_i$ for all $i\in[k]$\;
    Let $T_{i-1}'$ be the result of contracting $\C(S_{i-1,j})$ for all $S_{i-1,j}\in \mathcal{S}_{i-1}$\;
    Let $\mathcal{L}_{i-1}$ be the set of all maximal leaf-stars (containing their parent) in $T_{i-1}'$\;
    Let $T_i$ be the result of contracting $\C(L_{i-1,j})$ for all $L_{i-1,j}\in \mathcal{L}_{i-1}$\;
  }
  \KwRet{$\C(T_l)$}\;
 \caption{\textsf{BoundedTreeContract} \\(Computing the solution $P(T)$ of a problem $P$ on degree-weighted tree $T$ with max degree $n^\epsilon$ using connected contracting function $\mathcal{C}$)}
\label{alg:boundedtreecontract}
\end{algorithm}

Notice that we can view the first and second steps as the \compress{} stage except that we limit each component such that the sum of the degrees in each component is at most $n^\epsilon$. Since the size of the vector $W(v)$ is $w(v) = \widetilde{O}(\deg(v)) = \widetilde{O}(n^\epsilon)$, we can store an entire component (in its current, compressed state) in a single machine, thus making the second step distributable. The third step can be viewed as a \rake{} function which, as we stated, can be handled on one machine per contraction using the connected contracting function.

In order to get $O(1/\epsilon^2)$ rounds, we first would like to show that the number of phases is bounded by $O(1/\epsilon)$. To prove this, we show that there will be at most one non-leaf node after we contract the components in each group. In other words, the dependency tree resulting from the preorder decomposition has at most one non-leaf node per group in the decomposition. This is a necessary property of decomposing the tree based off the preorder traversal. To see why this is true, consider a connected component in a partition. If it is \textit{not} the last connected component (i.e., it does not contain the partition's last vertex according to the preorder numbering), then after contracting, it cannot have any children. 

\begin{lemma}\label{lemma:preorder-decomposition}
The dependency tree $T'=(V',E')$ of a preorder decomposition $V_1, V_2, \ldots, V_k$ of tree $T=(V,E)$ contains at most $1$ non-leaf vertex per group for a total of at most $k$ non-leaf vertices. In other words, there are at most $k$ dependent connected components in $\cup_{i\in[k]}{F_i}$.
\end{lemma}
\begin{proof}
Each group $V_i$ induces a forest $F_i$ on tree $T$, and recall that each $F_i$ is consisted of multiple connected components $C_{i,1}, C_{i,2}, \ldots, C_{i, c_i}$, where $c_i$ is the number of connected components of $F_i$. Assume w.l.o.g. component $C_{i,c_i}$ is the component which contains vertex $l_i$, the vertex with the largest index in $V_i$. We show that every connected component in $F_i$ except $C_{i,c_i}$ is independent, and thus Lemma~\ref{lemma:preorder-decomposition} statement is implied. See in Subfigure~\ref{fig:dependency} that there is at most $1$ dependent component, red vertices in $T'$, for each group $V_i$. Also note that $C_{i,c_i}$, the only possibly dependent component in $F_i$, is always the last component if we sort the components based on their starting index since $l_i \in C_{i,c_i}$ and each $C_{i,j}$ contains a consecutive set of vertices.

Assume for contradiction that there exists $C_{i,j}$ for some $i \in [k], j \in [c_i-1]$, a non-last component in group $V_i$, such that $C_{i,j}$ is a dependent component, or equivalently $C'_{i,j}$ is not a leaf in $T'$. 
Since $C_{i,j}$ is a dependent component, there is a vertex $v \in C_{i,j}$ which has a child outside of $V_i$. Let $u$ be the first such child of $v$ in the pre-order traversal, and thus $u \in V_j$ for some $j > i$. Consider a vertex $w \in V_i$ that comes after $v$ in the pre-order traversal.
Then, since $u$, and thus $V_j$, comes after $v$ and $V_i$ in the pre-order traversal, $u$ must come after $w$ in the pre-order traversal. Since $w$ is between $v$ and $u$ in the pre-order traversal, and $u$ is a child of $v$, the only option is for $w$ to be a descendant of $v$. Then the path from $w$ to $v$ consists of $w$, $w$'s parent $p(w)$, $p(w)$'s parent, and so on until we reach $v$. Since a parent always comes before a child in a pre-order traversal, all the intermediate vertices on the path from $w$ to $v$ come between $w$ and $v$ in the pre-order traversal, so they must all be in $V_i$. This means $w$ is in $C_{i,j}$ since $w$ is connected to $v$ in $F_i$. Since any vertex after $v$ in $V_i$ must be in $C_{i,j}$, $C_{i,j}$ must be the last connected component, i.e., $j=c_i$. This implies that the only possibly dependent connected component of $F_i$ is $C_{i,c_i}$, and all other $C_{i,j}$'s for $j \in [c_i-1]$ are independent.
\end{proof}

Lemma~\ref{lemma:preorder-decomposition} nicely fits with our result from Observation~\ref{obs:groups} to bound the total number of phases \textsf{BoundedContract} requires. In addition, we can show how to implement each phase to bound the complexity of our algorithm. Note that we are assuming that our component contracting function is defined to always yield a degree-weighted tree. We only need to show that the degrees stay bounded throughout the algorithm.

\begin{proof}[Proof of Theorem~\ref{thm:boundedmain}]
In each phase of this algorithm, the only modifications to the graph that occur are applications of the connected contracting functions to connected components of the tree. Since these are assumed to preserve $P(T)$ and we simply solve $P(T_l)$ for the final tree $T_l$, correctness of the output is obvious.  

An important invariant in this algorithm is the $O(n^\epsilon)$ bound on the degree of vertices throughout the algorithm. At the beginning, we know that the degrees are bounded according as it is promised in the input. We show that this bound on the maximum degree of the tree is invariant by proving the degree of vertices are still bounded after a single contraction. 

Recall that we use preorder decomposition with $\lambda=n^\epsilon$ to find the connected components we need to contract in the \compress{} stage. According to definition, the total degree of each group in our decomposition is bounded by $\lambda$. After we contract a component $S$, the degree of the contracted vertex $v^S$ never exceeds the sum of the degree of all vertices in $S$ since every child of $v^S$ is a child of exactly one of the vertices in $S$. Thus, the degree of $v^S$ is bounded by $\lambda = n^\epsilon$. In \rake{} stage, we contract a number of sibling leaves into their common parent. In this case, the degree of the parent only decreases and the bound still holds.

We now focus on round and space complexities. A preordering can be computed using the preorder traversal algorithm from Behnezhad et al.~\cite{behnezhad2019massively}, which executed in $O(1/\epsilon)$ rounds with $O(n^\epsilon)$ local space and $\widetilde{O}(n)$ total space w.h.p.\footnote{This means with probability at least $1/poly(n)$} 
This completes step 1. In steps 2 and 3, the contracting functions are applied in parallel for a total of $O(1)$ rounds (based off our assumption about any given contracting functions) within the same space constraints. Thus, all phases require $O(1)$ rounds except the first, which is $O(1/\epsilon)$ rounds, and satisfy the space constraints of our theorem.

Now we must count the phases. Lemma~\ref{lemma:preorder-decomposition} tells us that for every group, we only have one non-leaf component in the dependency graph after each step 2. In step 3, we then ``\rake{}'' all leaves into their parents. This means that the remaining number of vertices after step 3 is equal to the number of non-leaf vertices in the dependency graph after step 2, which is $k=n^\epsilon$. Observation~\ref{obs:groups} tells us that the resulting graph size is then $O(n/n^\epsilon) = O(n^{1-\epsilon})$. Therefore, in order to get a graph where $|T_l|=1$, we require $O(1/\epsilon)$ phases. Combining this with the complexity of each phase yields the desired result.
\end{proof}

%% file: simulate.tex
\subsection{Simulating $2$-tree-contraction in $O(1)$ AMPC rounds} \label{sec:2contract}
Due to Theorem~\ref{thm:main}, we can compute any $P(T)$ on trees as long as we are provided with a connected contracting function and a sibling contracting function with respect to $P$. A natural question that arises is the following: for which class of problem $P$ there exists black-box contracting functions? We argue that many problems $P$ for which we have a $2$-tree-contraction algorithm can also be computed in $O(1/\epsilon^3)$ AMPC rounds using $n^\epsilon$-tree-contraction.

In many problems which are efficiently implementable in the Miller and Reif~\cite{miller1989parallel} \emph{Tree Contraction} framework, we are given $\C$ and $\R$ contracting functions, for \blackbox{Compress} and \blackbox{Rake}  stages respectively, which contract only one node: either a leaf in case of \blackbox{Rake} or a vertex with only one child in case of \blackbox{Compress}. Let us call this kind of contracting functions \emph{unary contracting functions} and denote them by $\C^1$ and $\R^1$. This is a key point of original variants of Tree Contraction which  contract odd-indexed vertices, or contract a maximal independent set of randomly selected vertices. Working efficiently regardless of using only \emph{unary} contracting functions is the reason Tree Contraction was considered a fundamental framework for designing parallel algorithms on trees in more restricted models such as PRAM. For example, in the EREW variant of PRAM, an $O(\log(n))$ rounds tree contraction requires to use only \emph{unary} contracting functions $\R^1$ and $\C^1$. More generally, we define $i$-ary contracting functions as follows.

\begin{definition}
An \emph{``$i$-ary contracting function''}, denoted by $\C^i$ or $\R^i$, is a \emph{contracting function} which admits a subset $S = \{ v_1, v_2, \ldots, v_k \}$ of at most $i+1$ vertices at a time such that $\sum_{j=1}^k{deg(v_j) = O(i)}$. A special case of $i$-ary contracting functions, are \emph{``unary contracting functions''}, denoted by $\C^1$ or $\R^1$, which contract only one vertex at a time.
\end{definition}

However, in the the AMPC model, we can contract the chains more efficiently, and thus we are allowed to utilize more relaxed variants of \blackbox{Compress} stage. Furthermore, as we show in Theorem~\ref{thm:dptwo}, designing \emph{unary} contracting functions $\C^1$ and $\R^1$ is not easier than designing \emph{$i$-ary} contracting functions $\C^i$ and $\R^i$ in the AMPC model. We show this by reducing $\C^i$ and $\R^i$ to $\C^1$ and $\R^1$ in $O(1)$ rounds for any $i = O(n^\epsilon)$. In other words, the restrictions of PRAM model, which requires $\C^1$ and $\R^1$ exclusively, enables us to directly translate a vast literature of problems solved using tree contraction to efficient AMPC algorithms for the same problem given $\C^1$ and $\R^1$. 

As we have shown in Theorem~\ref{thm:main}, it is possible to solve any problem $P(T)$ in $O(1/\epsilon^3)$ AMPC rounds given a connected contracting function $\C$ and a sibling contracting function $\R$, where both are $n^\epsilon$-ary contracting with respect to $P$. In what follows, we demonstrate the construction of $n^\epsilon$-ary contracting functions given a unary connected contracting function $\C^1$ and a sibling contracting function $\R^1$. 

\begin{restatable}{theorem}{dptwo}
\label{thm:dptwo}
Given a  \emph{unary connected contracting function} $\C^1$ and a \emph{unary sibling contracting function} $\R^1$ with respect to a problem $P$ defined on trees, one can build an \emph{$i$-ary connected contracting function} $\C^i$ and an \emph{$i$-ary sibling contracting function} $\R^i$ with respect to $P$ and both $\C^i$ and $\R^i$ run in one AMPC rounds as long as $i = O(n^\epsilon)$.
\end{restatable}



\begin{proof}
First, we present an algorithm for $\C^i$. We are given a connected subtree induced by $S = \{ v_1, v_2, \ldots, v_k \}$ of $T$ so that $\sum_{j=1}^k{deg_{T}(v_j)} = O(i)$. Since $i = O(n^\epsilon)$, the whole subtree fits into the memory of a single machine. Some of the leaves of this subtree are \emph{known}, meaning that they are a leaf also in $T$, and others are \emph{unknown}, meaning that they have children outside $S$. Let $\mathcal{U} = \{ u_1, u_2, \ldots, u_l \}$ be the set of the children of \emph{unknown} leaves as well as the children of non-leaf nodes  which are outside of $S$. Ultimately, we want to compress the data already stored on the vertices of $S$ into a memory of $\widetilde{O}(l)$ as the degree of $v_1$ in the contracted tree $T'$ will be $l + 1$, and thus $deg_{T'}(v_1) = l + 1$. 

The $W(v_1)$ of each contracted vertex $v_1$ is a weighted-degreee tree structure $T^{\C}(v_1)$ whose leaves are the children of $v_1$ in $T'$, and there is no vertex with exactly one child in $T^{\C}(v_1)$. Thus, the number of vertices in $T^{\C}(v_1)$ is bounded by $O(l) = O(deg_{T'}(v_1))$. In addition, we are guaranteed that the total size of vectors on each vertex of $T^{\C}(v_1)$ is bounded by $|T^{\C}(v_1)|$  since $\C^1$ and $\R^1$ are unary contracting functions. 
Therefore, we assume each $W(v_j)$ for each vertex $v_j \in S$  has stored a tree structure of size $\widetilde{O}(deg_T(v_j))$. We concatenate all these trees to get an initial $T^{\C}(v_1)$ whose size is bounded by $\sum_{j=1}^k{\widetilde{O}(deg_T(v_j))} = \widetilde{O}(n^\epsilon)$.

We run a $2$-tree-contraction-\emph{like} algorithm locally on $T^{\C}(v_1)$ using $\C^1$ and $\R^1$. Note that we can only rake the known leaves since the data of unknown leaves depend on their children. We repeating \compress{} and \rake{} stages until there is no known leaf or a vertex with one child remain in $T^{\C}(v_1)$. Then, according to Lemma~\ref{lemma:rounds} for $\alpha=2$, the number of remaining vertices in $T^{\C}(v_1)$ is bounded by $O(l)$. We store the final $T^{\C}(v_1)$ in $T'$ which requires a memory of $\widetilde{O}(l) = \widetilde{O}(w_{T'}(v_1))$. Hence, $\C^i$ satisfies the size-constraint on the weight vectors of the resulting weighted-degree tree.

Finally, we present an algorithm for $\R^i$ which is more straight-forward compared to that of $\C^i$. We are given a leaf-star $S = \{ v_1, v_2, \ldots, v_k \}$ of $T$ so that $\sum_{j=1}^k{deg_T(v_j)} = O(i)$. This implies that there are at most $O(n^\epsilon)$ vertices in $S$ as long as $i = O(n^\epsilon)$, and we can fit the whole $S$ into a memory of a single machine. To simulate $\R^i$, we only need to $k-1$ times apply $\R^1$ on $S_j = \{ v_1, v_{j+1} \}$ at the $j$-th iteration. Note that every $v_j \in S$ is a  leaf in $T$, so the data stored in $W(v_j)$ is just $\widetilde{O}(1)$ bits and not a tree structure. Theorem~\ref{thm:dptwo} statement is implied.
\end{proof}

\subsection{Reconstructing the Tree for Linear-sized Output Problems} \label{sec:reconstruct}

Consider a problem $P(T)$ whose output size is also linear in the size of input $n$. For instance, in maximum weighted matching (which we thoroughly study in \S\ref{sec:maximummatching}) we need to find the matching itself. Up to this point, in all of our algorithms, we assume the output of function $P(T)$ is of constant-sized. We simply contract the tree through some iterations until it collapses into a single vertex, and we do not need to remember anything about a vertex which is contracted as a member of a connected component or as a member of a leaf-star. 

In this section, we present a general approach for retrieving the linear-sized solution in a natural scenario, where we need to retrieve a recursively-defined weight vector $P(v)$ of constant size for each vertex $v \in V$. In the special case of maximum weighted matching which can be formulated as a dynamic programming problem, $P(v)$ contains the final value of different DP values with respect to the subtree rooted at $v$\footnote{Note that retrieving $P(v)$ for each vertex $v$ still does not give us the optimum matching and a problem-specific post-processing step is required to retrieve the actual matching}. 

Roughly speaking, our reconstruction algorithm is based on storing the information about components we contracted throughout the algorithm in an auxiliary memory of size $O(n)$. It is easy to observe that if we store the degree-weighted subtree of every connected component or leaf-star that we contract during the algorithm we need at most $O(n)$ addition total memory. Note that during each application of black-box contracting functions, we remove at least one vertex from the tree and each vertex except root is removed exactly once when the algorithm terminates. Namely, for every phase $i$ we need to store $\mathcal{S}_i$ and $\mathcal{L}_i$ in Algorithm~\ref{alg:boundedtreecontract}, and $\mathcal{K}_i$ and every $L_{i,j,z}$  in Algorithm~\ref{alg:main} (In addition to the data stored by each black-box application of Algorithm~\ref{alg:boundedtreecontract}). Since we have adaptive access to these subsets in AMPC, it is sufficient to index them by the id of the surviving vertex of each subset.

The full reconstruction algorithm starts after the main contraction algorithm finishes. We only need to store the information about contracted subsets during the running time of the contraction algorithm. Next, we iterate over the phases of the algorithm in reverse order, i.e., $i = \{ 1/\epsilon, 1/\epsilon - 1, \ldots, 1 \}$, and undo the contractions that were performed during phase $i$. Let $C_v$ be a connected contracted component rooted at $v$, and $\{ w_1, w_2, \ldots, w_k \}$ be children of $v$ post-contraction. 

Whenever we undo a connected contraction like $C_v$, we replace $v$ with the whole structure of $C_v$ including $W(u)$ for every $u \in C_v, u \neq v$. Then we populate the $P(u)$ for every $u \in C_v, u \neq v$. During the contraction algorithm $P(w_j)$ is not known for any $j$. However, during the reconstruction we know $P(w_j)$ for every $1 \leq j \leq k$ since these vertices are contracted in a later phase than the phase we contract $C_v$. Hence, we have already populated $P(w_j)$ and we can use these values to locally populate $P(u)$ for every $u \in C_v$. Undoing the sibling contracting functions in much simpler since their values do not depend on other vertices nor the value of other vertices depend on their value. We populate $P(u)$ for every $u \in L$, where $L$ is a leaf-star, based on the already constant-sized weight vectors $W(u)$. 



%% file: applications.tex
\section{Applications}\label{sec:applications}

In this section, we show that our tree contraction algorithm (Algorithm~\ref{alg:main}) as well as our bounded contraction algorithm (Algorithm~\ref{alg:boundedtreecontract}) extend to a number of applications addressed by both Miller and Reif~\cite{miller1985parallel} and Behnezhad et al.~\cite{bateni2018massively}. While this covers many important examples of solving problems in AMPC using tree contractions, this is only a subset of the problems we can solve.

\applications*

This section has multiple subsections that show different strategies to solve problems using Algorithms~\ref{alg:boundedtreecontract} and~\ref{alg:main}. Across the sections, we prove a lemmas for each problem in Theorem~\ref{thm:applications} that ultimately prove the theorem. In Section~\ref{sec:maximummatching}, we apply Algorithm~\ref{alg:main} to solve maximum weighted matching and maximum weighted independent set in $O(1/\epsilon^3)$ rounds. In Section~\ref{sec:maximalmatching}, we apply Algorithm~\ref{alg:boundedtreecontract} to solve maximal matching and maximal independent set in $O(1/\epsilon^2)$ rounds. Finally, in Section~\ref{sec:treeiso}, we use Algorithm~\ref{alg:boundedtreecontract} to solve expression evaluation in $O(1/\epsilon^2)$ rounds and show how this extends to a similarly efficient algorithm for tree isomorphism testing.

\subsection{Maximum Weighted Matching and Independent Set}\label{sec:maximummatching}

In this section, we show how to solve maximum weighted matching and independent set on trees efficiently in AMPC. While these results are explained in the context of maximum weighted matching, the same strategies translate to maximum weighted independent set. For maximum weighted matching, or MWM, we prove:

\begin{lemma}\label{lem:mwm}
MWM on trees can be solved in $O(1/\epsilon^3)$ AMPC rounds with $O(n^\epsilon)$ local memory and $\widetilde{O}(n)$ total memory.
\end{lemma}

Consider a weighted tree $T = (V,E, \omega)$ where $\omega$ is the edge weight function. We will let $T_v$ for any vertex $v\in V$ be the maximal subtree rooted at $v$ in $T$. We first explain how a standard bottom-up dynamic program would solve this problem. Let $D$ be the dynamic program table where $D(v)$ for any $v\in V$ is the stored data corresponding to $T_v$. This dynamic program gets solved inductively based off vertex height starting at the leaves. For each vertex $v\in V$, we want to find two main values and store them in $D(v)$: (1) $c_v$, the value of the MWM on $T_v$, and (2) $c_v'$, the same value except we do \emph{not} allow $v$ to be matched. If $Ch_v$ is the children of $v$, we can write $c_v$ and $c_v'$ in terms of that of its children:

\begin{align*}
c_v =& \max\left\{\max_{u\in Ch_v} \left(\omega(u,v) + c_u' + \sum_{x\in Ch_v\setminus\{u\}} c_x\right), \sum_{u\in Ch_v} c_u\right\}\tag{1}\label{cv}
\\c_v' =& \sum_{u\in Ch_v} c_u\tag{2}\label{cv'}
\end{align*}

For the simpler $c_v'$, we know that the value of the MWM of $T_v$ assuming $v$ is not matched is the sum of the MWMs on $T_u$ for all $u\in Ch_v$. For $c_v$, this is one possibility, but we could also consider matching $v$ to some $u\in Ch_v$. In that case, we get the added edge weight $\omega(u,v)$, but we require $u$ to not be matched when we consider the MWM of $T_u$, so we use $c_u'$ instead of $c_u$. Then $c_v$ is just the maximum of all these options. We can simplify these equations by pulling out the summations:

\begin{align*}
c_v =& \max\left\{\max_{u\in Ch_v} \left(\omega(u,v) + c_u' - c_u\right), 0\right\}\tag{3} + \sum_{u\in Ch_v} c_u\label{cv2}
\end{align*}

This is the standard dynamic program to solve MWM on trees. The data computed at each vertex $v$ is $D(v) = (c_v, c_v')$. This is passed up to its parent $u$ so that $D(u) = (c_u, c_u')$ can be solved and so on until we are able to compute $c_r$ where $r$ is the root of $T$. This will be our final solution that computes the value of the maximum weighted matching.

Before we describe our contraction process, we formally define our dynamic program. Note that this will include some notation not introduced in the description above. Afterwards, we will describe why the additional pieces of information are necessary.

\begin{definition}\label{def:dp}
Let $D$ be the MWM dynamic table on a given tree $T = (V,E)$ that acts on both vertices and edges. Then $D(v)= (c_v, c_v', a_v, b_v)$ and $D(e) = (\omega_1(e),\omega_2(e),\omega_3(e),\omega_4(e))$ where:
\begin{itemize}
\item $c_v$ is a function that computes the MWM on $T_v$
\item $c_v'$ is a function that computes the MWM on $T_v$ assuming $v$ is not matched
\item $a_v$ and $b_v$ are constants
\item $\omega_i(e)$ for $i\in[4]$ are four different constant weights on edge $e$
\end{itemize}
and more specifically:
\begin{align*}
c_v =& \max\{\max_{u\in Ch_v} \left(\omega_1(u,v) + c_u' - c_u\right),
\max_{u\in Ch_v} \left(\omega_2(u,v)\right),
\max_{u\in Ch_v} \left(\omega_3(u,v) + c_u' - c_u\right),
\\&\qquad\,\,\max_{u\in Ch_v} \left(\omega_4(u,v)\right),
a_v, 0\} + \sum_{u\in Ch_v} c_{v,u} + b_v
\\c_v' =& \sum_{u\in Ch_v} c_{v,u} + b_v
\end{align*}
for $c_{v,u} = \max\{\omega_3(u,v) + c_u', \omega_4(u,v) + c_u\}$.
\end{definition}

It is not too hard to see that this is a generalization of the MWM dynamic program. If we let $\omega(e) = (\omega(v), -\infty, -\infty, 0)$ for all $e\in E$ and $a_v=b_v=0$ for all $v\in V$, then this directly \textit{becomes} the dynamic program we mentioned before. Therefore:

\begin{proposition}\label{prop:dp}
Consider a tree $T = (V,E)$. The dynamic program $D$ from Definition~\ref{def:dp} solves MWM when $\omega(e) = (\omega(v), -\infty, -\infty, 0)$ for all $e\in E$ and $a_v=b_v=0$ for all $v\in V$. In other words, for all $v\in V$, $c_v$ from $D(v)$ is the value of the MWM on $T_v$.
\end{proposition}

The reason why we require this additional data in $D$ to implement our algorithm is that it cannot compute this dynamic program simply from bottom up. In intermediate steps, it will have to contract arbitrary connected components $C$ rooted at $v$ into a single vertex. Since $C$ should represent the value of the MWM of the maximal subtree containing all of the component $C$, we let $c_C = c_v$ and $c_C'=c_v'$. That way, the maximal subtree in question, $T_v$, contains the entire connected component $C$.



In order to solve the dynamic program on the tree after contracting $C$, we need to be able to compute $c_v$ and $c_v'$ as functions of the data of $C$'s children, $\{D(u)\}_{u\in Ch_C}$, after $C$ is contracted. This is precisely the information we need to encode into $C$: how to compute $c_C = c_v$ and $c_C'= c_v'$ in terms of $\{c_u, c_u'\}_{u\in Ch_C}$. Obviously, this could be easily done by recalling the entire structure of $C$ and using that to continue the dynamic program from $C$'s children up through the internal nodes of $C$ until we compute $v$.  However, this is inefficient, because this is as large as $|C|$. Recall that our algorithm must store only $\widetilde{O}(\deg(C))$ bits, or alternatively, $O(\deg(C))$ values. Instead, we show how to contract the component $C$ to create a smaller component $C'$ where $|C'| = O(\deg(C))$ such that we can still solve the dynamic program on $C$.

First, we observe that $|C'| = O(\deg(C))$ if it contains no internal leaves (i.e., vertices that have no children in $T$) and no internal nodes that only have one child in $T$. We show how to contract such vertices to reduce the component size. Consider a set of leaves $\mathcal{L}\in C$ that share a parent $p$. Since $C$ is connected, $p\in C$. Since any $\ell\in \mathcal{L}$ has no children in or out of $C$, we know $c_\ell$ and $c_\ell'$ have no dependencies, meaning they are constants. Consider rewriting the equation for $c_p$ and $c_p'$ from Definition~\ref{def:dp} by simply separating the leaf and nonleaf children:

\begin{align*}
c_p =& \max\{\max_{u\in Ch_p\setminus\mathcal{L}} \left(\omega_1(u,p) + c_u' - c_u\right),
\max_{u\in Ch_p\setminus\mathcal{L}} \left(\omega_2(u,p)\right),
\\&\qquad\,\,\max_{u\in Ch_p\setminus\mathcal{L}} \left(\omega_3(u,p) + c_u' - c_u\right),
\max_{u\in Ch_p\setminus\mathcal{L}} \left(\omega_4(u,p)\right),
\\&\qquad\,\,\max\{\max_{\ell\in\mathcal{L}} \left(\omega_1(u,p) + c_u' - c_u\right), \max_{u\in Ch_p\setminus\mathcal{L}} \left(\omega_2(u,p)\right),
\\&\qquad\qquad\max_{u\in Ch_p\setminus\mathcal{L}} \left(\omega_3(u,p) + c_u' - c_u\right),
\max_{u\in Ch_p\setminus\mathcal{L}} \left(\omega_4(u,p)\right),a_p\}, 0\} \\&+ \sum_{u\in Ch_p\setminus\mathcal{L}} c_{p,u} + \sum_{\ell\in \mathcal{L}}  c_{p,\ell} + b_p
\\c_p' =& \sum_{u\in Ch_p\setminus\mathcal{L}} c_{p,u}+ \sum_{\ell\in \mathcal{L}} c_{p,\ell}  + b_p.
\end{align*}

Note that many of these terms, only consist of constants. Specifically the nested maximization term in $c_p$ and $\sum_{\ell\in \mathcal{L}} c_{p,\ell}  + b_p$, which appears in both $c_p$ and $c_p'$, are constants. When we contract, we can compute these two constant values and combine them with the constant values associated with $D(p)$. This is how we compute $a_p$ and $b_p$ from Definition~\ref{def:dp}. Note that these are running values that change over time. Specifically, $a_p$ and $b_p$ start at zero (as in the standard MWM problem), and as a vertex $\ell$ that is a child of $p$ gets trimmed:
\begin{align*}
a_p \gets& \max\{\max_{\ell\in\mathcal{L}} \left(\omega_1(u,p) + c_u' - c_u\right), \max_{u\in Ch_p\setminus\mathcal{L}} \left(\omega_2(u,p)\right),
\max_{u\in Ch_p\setminus\mathcal{L}} \left(\omega_3(u,p) + c_u' - c_u\right),
\\&\quad\,\,\,\,\max_{u\in Ch_p\setminus\mathcal{L}} \left(\omega_4(u,p)\right),a_p\}
\\b_p \gets& \sum_{\ell\in \mathcal{L}} c_\ell + b_p
\end{align*}

Now $c_p$ and $c_p'$ are no longer functions of $\mathcal{L}$. Therefore, we can safely trim all leaves. This is the first step in simplifying $C$ to make $C'$.

Next, we consider vertices with one child. More generally, let a \textit{maximal chain} in $C$ be a maximal path of edges from parents to children where the children have one child, and we additionally include the descending edge from the final child. Consider some maximal chain $P = (e_1,\ldots, e_k)$. We will contract this into a single edge $e_P$. Recall that $D$ stores a four-tuple $D(e_p) =  (\omega_1(e_P),\omega_2(e_P),\omega_3(e_P),\omega_4(e_P))$, which is necessary for the maximal chain contractions. Here, each value represents the value of the MWM along the path under the following restrictions: $\omega_1(e_P)$ represents that $e_1$ and $e_k$ are matched, $\omega_2(e_P)$ represents $e_1$ is matched but $e_k$ is not matched, $\omega_3(e_P)$ represents $e_1$ is not matched but $e_k$ is matched, and $\omega_4(e_P)$ represents neither $e_1$ nor $e_k$ are matched. For example, consider when our path is just one edge $e_1$. Matching across $e_1$ yields $\omega(e_1)$ value, so $\omega_1(e_P) = \omega(e_1)$. The parent cannot match without the child being matched, so $\omega_2(e_P) = -\infty$, and similarly $\omega_3(e_P) = -\infty$. If neither parent nor child are allowed to match, the MWM yields 0 value. Thus $D(e_P) = (\omega(e_1), -\infty, -\infty, 0)$. This is why setting the edge values in the dynamic program in this way reduces the problem to MWM.

On a larger path, however, these values might become more general. In our chain, let a child, parent, and grandparent vertices be $c, p$ and $g$ respectively, with edges $e_k=(c,p)$ and $e_{k-1}=(p,g)$. Since we are contracting this into an edge $e = (g,c)$, the indices need to represent: matching both $g$ and $c$ along the path, matching just $g$ along the path, matching just $c$ along the path, and matching neither $g$ nor $c$ along the path. We can do this in three ways: (1) by matching both $c$ and $p$ along $e_k$ and just matching $g$ along $e_{k-1}$, (2) by matching just $c$ along $e_k$ and matching both $p$ and $g$ along $e_{k-1}$, or (3) by matching just $c$ along $e_k$ and just $g$ along $e_{k-1}$. In (1) and (2), $p$ is matched along one of the paths, thus it cannot match to any of its other children that may have existed in an earlier iteration of the tree. Thus we can only aggregate $c_{p\setminus c}'$ into this maximum matching, where $c_{p\setminus c}'$ (resp. $c_{p\setminus c}$) is the same as $c_p'$ (resp. $c_{p\setminus c}$) but assuming we simply cut $p'$ from its child $c$. In (3), however, $p$ is free to match with another of its children, thus we can use $c_p$. This shows how to compute part of $\omega_1(e_P)$. We can compute the other three values in a similar way. By using this process, we find:

\begin{align*}
\omega_1(e_{k-1},e_k) =& \max(\omega_1(e_k)+\omega_2(e_{k-1})+c'_{p\setminus c}, \omega_3(e_k)+\omega_1(e_{k-1})+c'_{p\setminus c}, \\&\qquad\,\,\,\omega_3(e_k)+\omega_2(e_{k-1})+c_{p\setminus c})
\\\omega_2(e_{k-1},e_k) =& \max(\omega_2(e_k)+\omega_2(e_{k-1})+c'_{p\setminus c}, \omega_4(e_k)+\omega_1(e_{k-1})+c'_{p\setminus c}, \\&\qquad\,\,\,\omega_4(e_k)+\omega_2(e_{k-1})+c_{p\setminus c})
\\\omega_3(e_{k-1},e_k) =& \max(\omega_1(e_k)+\omega_4(e_{k-1})+c'_{p\setminus c}, \omega_3(e_k)+\omega_3(e_{k-1})+c'_{p\setminus c},\\&\qquad\,\,\, \omega_3(e_k)+\omega_4(e_{k-1})+c_{p\setminus c})
\\\omega_4(e_{k-1},e_k) =& \max(\omega_2(e_k)+\omega_4(e_{k-1})+c'_{p\setminus c}, \omega_4(e_k)+\omega_3(e_{k-1})+c'_{p\setminus c}, \\&\qquad\,\,\,\omega_4(e_k)+\omega_4(e_{k-1})+c_{p\setminus c})
\end{align*}

Notice that all these values are constant since all edge weights are known and $c_{p\setminus c}$ and $c_{p\setminus c}'$ must be known since $c$ is $p$'s only child. Repeatedly applying this to the bottom two adjacent edges eventually contracts all edges and leaves us with $\omega_i(e_k)$ for all $i\in[4]$. It is not hard to see that, assuming inductively that all computed weights and $c(v)$ and $c'(v)$ values are correct throughout this entire process, then the final weight tuple $\omega(e_p)$ is correct after contracting the maximal chain. After contracting all maximal chains to form $C'$ and then trimming leaves with the aforementioned process, $C'$ becomes a tree with no degree 1 vertices and no internal leaves. Therefore, if $\mathcal{L}_{C'}$ is the leaves of $C'$, $|C'| = O(|\mathcal{L}_{C'}|)$. Since they are not internal leaves, they must each have at least one child outside of $C'$. Therefore, this is a lower bound on $\deg(C')$. Thus $|C'| = O(\deg(C')) = O(\deg(C))$.

This shows how to contract $C$ into $C'$ such that $|C'| = O(\deg(C))$ where we still have the information to compute $c_v$ and $c_v'$. Note that this is how we will store $c_C$ and $c_C'$ in $D(C)$: as a component $C'$ with $|C'| = O(\deg(C))$ such that the values of $c_v$ and $c_v'$ in terms of $Ch_C$ are the same as they were in the original component $C$. Since we will refer to this process, we will create a formal definition for it:

\begin{definition}\label{def:contract}
We call the process defined above, for replacing a component $C$ with a component $C'$, the \textbf{Connected Contraction Process}.
\end{definition}

As we have shown above:

\begin{proposition}\label{prop:contractprocess}
The Connected Contraction Process replaces any connected component $C$ with a component $C'$ such that $|C'| = O(\deg(C))$ and $c_v$ and $c_v'$ remain the same in $C$ and $C'$.
\end{proposition}

Now we can introduce the version of MWM on degree-weighted trees using the dynamic program formalization. The weight vector $W$ will simply be the dynamic program information. Note that the dynamic program stores data on edges as well, however, we can simply store this on the child vertex of each edge. Therefore, for all $v\in V$ with parent $p\in V$, $W(v) = (D(v), D(v,p))$. To show this satisfies the degree-weighted property, we will simply need to show that $\dim(W(v)) =\widetilde{O}(\deg(v))$.

\begin{definition}
Consider a \degwei{} tree $T = (V, E, W)$ where we have $W(v) = (D(v), D(v,p))$ for all $v\in V$ with parent $p\in V$ and $W(r) = (D(r))$ for root $r$, where all $D$ values are stored as explained in the Connected Contraction Process. Then the \textbf{degree weighted maximum weighted matching} on $T$ is equivalent to the problem of solving $D(v)$ for all $v\in V$.
\end{definition}

Then by extension of Proposition~\ref{prop:dp}, solving this MWM problem on \degwei{} trees can be used to solve MWM on standard trees with the aforementioned input $W$. Now we are ready to apply our main algorithms.

\begin{lemma}\label{lem:maxmatch}
Given a \degwei{} tree $T=(V,E,W)$, there exists an $O(1/\epsilon^3)$ round AMPC algorithm for finding the value of the maximum weighted matching in $T$. The memory per machine is $O(n^\epsilon)$, and the total memory is $\widetilde{O}(n)$.
\end{lemma}

\begin{proof}
First, we must describe our connected contracting function $\C$. Consider some connected component $C \subseteq V$. We simply apply the Connected Contraction Process to contract the component and define the new $W(C)$ weight vector. By Proposition~\ref{prop:contractprocess}, this ensures that $\dim(W(C)) = \widetilde{O}(\deg(v))$, thus this is a valid weight vector. It is also a valid connected contracting function, since $c_v$ and $c_v'$ as a function of $Ch_C$ does not change.

Next, we describe our sibling contracting function. This will also be quite similar to the Connected Contraction Process. One small nuance is we contract the sibling leaves into a single leaf, instead of contracting a contiguous component, such as the sibling leaves with its parent. It is not too hard to see that this can be done by putting a dummy vertex between the sibling leaves and its parent as an intermediate parent node, so that the parent is now the grandparent (note: this requires us to ensure matching with the dummy means that the parent will be matched, but this can be done with our four-tuple edge weights, $D(e)$). Then we contract the siblings into this dummy parent to create the new vertex. Again, by Proposition~\ref{prop:contractprocess}, we have a valid sibling contracting function.

Therefore, by Theorem~\ref{thm:main}, we can solve MWM on \degwei{} trees in $O(1/\epsilon^3)$ AMPC rounds with $O(n^\epsilon)$ memory per machine and $\widetilde{O}(n)$ total memory.
\end{proof}

Obviously then, since MWM on trees is a subproblem of MWM on \degwei{} trees, we can extend this solution to standard trees. Thus far, we have only shown how to compute $D(v)$ for each $v \in V$. In part, this gives us the value of the MWM of the maximal subtrees rooted at each vertex. We now show how to reconstruct the actual matching. Note to achieve this result, each vertex $v\in V$ must keep track of a pointer from it to its child that it selects in $c_v$, (i.e., the single child that is used in the maximization component in the computation of $c_v$, or no pointer if no child is selected). Therefore, assume each vertex is given such a pointer or possibly no pointer at all. We will refer to these pointers as child-match pointers.

\begin{lemma}\label{lem:reconstruct}
Given a tree $T = (V,E)$ along with $c_v$, $c_v'$, and the (possibly null) child-match pointer for each $v\in V$, we can find a MWM on $T$ in $O(1/\epsilon^3)$ AMPC rounds with $O(n^\epsilon)$ local space and $\widetilde{O}(n)$ total space.
\end{lemma}

\begin{proof}
Consider the tree where the edge set is reduced to just the edges along pointers. Since every vertex has at most one child, this must be a graph of disjoint paths. We know this must contain the MWM because for any $v$, the pointer from $v$ designates the child that maximizes $c_v$. In other words, it points to the child it must match to (if at all) such that the maximal subtree rooted at $v$ achieves a MWM. Therefore, finding a MWM on this disjoint paths graph will yield a maximum matching. Note that the root of each path (if viewed as a subforest of the original tree) tells us if the MWM on that component must match the root to its child. Therefore, a simple sequential algorithm would iterate over each path from top to bottom, recursively checking if we should match the root to its child or not (i.e., we do not match if and only if it matched to its parent already or if it has no child-match pointer). This would find the MWM, but since paths can be $O(n)$ in size, this is not directly distributable.

If a path is too long, simply segment it using the pre-order decomposition into segments of size $n^\epsilon$. For each segment, use this top-down algorithm to determine the MWM if the segment's root is not matched to its parent (i.e., the top-down algorithm as described), or if the root is matched to its parent (i.e., the top-down algorithm without the root vertex). Additionally, for both matches, store whether or not the last vertex in the segment gets matched. Then we simply contract the segments and repeat. Note that when we repeat on a path where nodes represent contracted segments, a super vertex could theoretically be matched above and below and have the matching be valid or possibly require it not be matched to either edge. This is accounted for by remembering if the last vertex is matched in the super vertex (for if it is, the descending edge cannot be selected) given we match the top edge. If we recursively apply this algorithm and then reverse the process, we will eventually achieve a maximum matching in $O(1/\epsilon)$ additional rounds for a total of $O(1/\epsilon^3)$ rounds.
\end{proof}

This is sufficient to prove Lemma~\ref{lem:mwm}.

\begin{proof}[Proof of Lemma~\ref{lem:mwm}]
Encode MWM on a standard tree $T = (V, E)$ as a \degwei{} tree by letting $W(v) = (c_v, c_v', 0, 0, w(v,p), -\infty, -\infty, 0)$ for every $v\in V$ with parent $p\in V$, or $W(r) = (c_r, c_r', 0, 0)$ for root $r$. We know this is equivalent to the standard MWM problem on $T$. By Lemma~\ref{lem:maxmatch}, we can evaluate $D$ at each vertex $v\in V$ in the required time and space. Then applying Lemma~\ref{lem:reconstruct} completes the proof.
\end{proof}

Again, this nicely translates into a maximum weighted independent set algorithm. In this case, instead of storing a 4-tuple of weights on edges based off of different ways to match along a path, we store such a weight that represents maximum independent sets along paths. This should yield an extremely similar connected contraction process to solve maximum weighted independent set.

\begin{lemma}\label{lem:mis}
Maximum weighted independent set on trees can be solved in $O(1/\epsilon^3)$ AMPC rounds with $O(n^\epsilon)$ local memory and $\widetilde{O}(n)$ total memory.
\end{lemma}

\input{mis}

\subsection{Dynamic Expressions and Tree Isomorphism}\label{sec:treeiso}
Miller and Reif~\cite{miller1991parallel} discuss how tree contractions can be used to probabilistically solve tree isomorphism. Let $T$ and $T'$ be rooted trees on $n$ nodes with tree height $h$ (it is not hard to see that height can be computed by tree contractions in $O(1/\epsilon^2)$ AMPC rounds by expanding high-degree vertices into $n^\epsilon$-ary trees with dummy nodes that don't add to the height). $T$ and $T'$ are isomorphic if there is a mapping between their vertices $\phi:V\to V'$  such that if $u$ is a  child of $v$ in $T$, then $\phi(u)$ is a child of $\phi(v)$.

A useful way to represent trees when considering tree isomorphism is as their canonically associated polynomials. On a tree $T$ with height $h$, we introduce $h$ variables $x_1,\ldots,x_h$ to define polynomials in $T$. For instance, for any vertex $\ell\in \leaves(T)$, let the polynomial associated with $\ell$ be $Q_\ell=1$. For any internal node $v$ of height $h_v$, we let $Q_v = \Pi_{u\in \children(v)} (x_{h_v} - Q_u)$. For some parameter $\alpha$, then, their algorithm, called \textsf{Randomized Tree Isomorphism}, is as follows:

\begin{enumerate}
\item If we are given a list of primes between $hn^{\alpha+1}$ and $2hn^{\alpha+1}$, then pick a prime in that range from the given list. Otherwise, pick a random integer in the range $(hn^{\alpha+1})^2\leq m \leq 2(hn^{\alpha+1})^2$.
\item For each node $v$ of $T$ or $T'$, assign a polynomial $Q_v$ to $v$. This is left in terms of the polynomials of its children.
\item Assign to each $x_i$ a random value between $1$ and $m$. 
\item Evaluate $Q_T$ and $Q_{T'}$ using dynamic expression evaluation and return $w$ and $w'$, respectively.
\item If $w\neq w'$, output ``not isomorphic'', else, output ``isomorphic''.
\end{enumerate}

They show the following:

\begin{theorem}[Miller \& Reif\cite{miller1991parallel}]
In the PRAM model, \textsf{Randomized Tree Isomorphism} tests tree nonisomorphsim in $O(\log n)$ time with $n/\log n$ processors with probability of error less than or equal to 1/2. If a table of primes is given, then the procedure works with a probability of error at most $1/n^\alpha$.
\end{theorem}

Their work also extends to canonical labelings for trees. We could extend our work in a similar way, but for the sake of only highlighting main results, we restrict our focus to tree isomorphism. We will implement this algorithm in AMPC and show:

\begin{lemma}\label{lem:treeiso}
Given two trees $T$ and $T'$ on $n$ nodes with height $h$, there exists an $O(1/\epsilon^2)$ round low-memory AMPC algorithm for determining nonisomorphism between $T$ and $T'$ with error probability less than $1/n^\alpha$ for some $\alpha$ when given access to a table of primes, and an error probability less than $1/2$ otherwise. The memory per machine is $O(n^\epsilon)$, and the total memory is $\widetilde{O}(n)$.
\end{lemma}

However, to show that this implementation works, we must describe how to execute dynamic expression evaluation. In this problem we are given a string of length $n$ which may include numbers, operators $+$, $-$, $\times$, $\div$, and $**$ (exponent), and parentheses. It must be a valid arithmetic expression. An example includes ``$2+5-(3+2\times 6)-9$''.

\begin{lemma}\label{lem:expression}
Given a $n$-length string-based representation of an arithmetic sequence involving $+$, $-$, $\times$, and $\div$ operators, there exists an $O(1/\epsilon^2)$ round low-memory AMPC algorithm for evaluating the expression. The memory per machine is $O(n^\epsilon)$, and the total memory is $\widetilde{O}(n)$.
\end{lemma}

\begin{proof}
The string preprocessing methods (i.e., forming a tree structure for expression evaluation) are all heavily inspired by  Bar-on et al.~\cite{baron1985optimal}, leveraging MPC to increase performance. We describe the process introduced by Bar-on et al. here. By extending their proof of correctness and showing an equivalent yet faster implementation in MPC, we can achieve our result. Their algorithm works as follows, modified to work efficiently in MPC:
\begin{enumerate}
\item Reduce the input string into a \textit{simple expression}, where all operations not in parentheses have the same precedence, and all maximal expressions in parentheticals are also simple:
\begin{enumerate}
\item For each $+$ and $-$ operator, insert two left parentheses to its right and two right parentheses to its left.
\item For each $*$ and $\div$ operator, insert one left parenthesis to its right and one right parenthesis to its left. 
\item For each left (resp. right) parenthesis, insert two additional left (resp. right) paretheses to its right (resp. left). Add two left parentheses to the beginning and two right parentheses to the end.
\end{enumerate}
\item Match parentheses (modified for the MPC model):
\begin{enumerate}
\item Partition the string into chunks of size $n^\epsilon$ and allocate one chunk to each machine.
\item On each machine, match all instances of ``$()$'' and remove them iteratively until none exist. The remaining string on any machine must be a sequence of right parentheses followed by left parentheses. Replace each sequence by a single appropriate parenthesis and the number of them.
\item Repeat this process. When matching appropriate right and left parentheses when they have an associated count, simply decrement the number under each parenthesis. At zero, remove the parenthesis.
\end{enumerate}
\item For each left (resp. right) parenthesis, delete itself if there is a left (resp. right) parenthesis to its right (resp. left) and their corresponding matched parentheses are adjacent.
\item Assign a processor to each subexpression (at this point, this is equivalent to assigning one to each pair of matched parentheses). The root is either the last or first operator, so find each and determine which is the root (this depends on operator precedence). Next, assign a processor to each operator. Look to adjacent operators and use the operator precedence to connect operators by a directed edge to denote which operator is the parent of the other.
\end{enumerate}

Step 1 can clearly be done in constant rounds, as each substep requires constantly many local operations. In Step 2 in MPC, we are iteratively reducing partitions of size $n^\epsilon$ to constant size. Thus this requires $O(1/\epsilon)$ rounds. Step 3, much like Step 1, is clearly local and requires $O(1)$ rounds. For Step 4, Bar-on et al. show that this can be done in $O(1)$ depth in PRAM by cleverly using the pointers between matched parentheses. We too can do this, and $O(1)$ depth in PRAM corresponds to $O(1)$ rounds in MPC. This takes a total of $O(1/\epsilon)$ rounds. It results in a binary tree with leaves that are values and internal nodes that are operators that correctly represents the order of operations.

To solve the tree, we use our tree contraction algorithm with data $W(v)$ storing the entire component with any solved vertices inputted into their parent. By extending our results from the previous section, we know that components $C$ that have no internal leaves or vertices with exactly one child are of the proper size to be stored in $W(C)$. This defines our compression function and allows us to utilize Algorithm~\ref{alg:boundedtreecontract}, which works on our tree because it is binary. Thus we can compute dynamic expression evaluation in $O(1/\epsilon^2)$ rounds.

\end{proof}

We can use this for our main lemma of the subsection.

\begin{proof}[Proof of Lemma~\ref{lem:treeiso}]
We show this by implementing \textsf{Randomized Tree Isomorphism} with the desired AMPC complexity. Step 1 can clearly be done in $O(1)$ rounds. Step 2 additionally only requires $O(1/\epsilon)$ rounds, as a vertex simply looks to its children and writes out its polynomial in terms of them. It requires $O(1/\epsilon)$ instead of $O(1)$ as, with $n^\epsilon$ local memory and possibly at most $O(n)$ children, a single polynomial may require $O(n^\epsilon)$ time to write. Step 3 also obviously takes $O(1)$ rounds. Finally, for Step 4, we simply defer to Lemma~\ref{lem:expression} to show this can be done in $O(1/\epsilon^2)$ rounds. Thus this takes a total of $O(1/\epsilon^2)$ rounds.
\end{proof}

%% file: mis.tex
\subsection{Maximal Matching and Independent Set}\label{sec:maximalmatching}

In this section, we show how to solve maximal matching and maximal independent set, or MIS, on trees in $O(1/\epsilon^2)$ rounds using Algorithm~\ref{alg:boundedtreecontract}. This is a surprising result because Algorithm~\ref{alg:boundedtreecontract} only works on trees with degree bounded by $n^\epsilon$. We do this by transforming instances of MIS on general trees into instances of a related problem, which we call \emph{maximal independent set with bypass vertices}, on trees of degree bounded by $n^\epsilon$. Then we solve this problem using Algorithm~\ref{alg:boundedtreecontract}. This section will be discussed in terms of MIS however the methods for maximal matching are similar.

\begin{lemma}\label{lem:ampcmis}
MIS on trees can be solved in $O(1/\epsilon^2)$ AMPC rounds with $O(n^\epsilon)$ local memory and $\widetilde{O}(n)$ total memory.
\end{lemma}

Recall that  Algorithm~\ref{alg:boundedtreecontract} requires an input that is a \degwei{} tree with degree bounded by $n^\epsilon$. To start, we reformulate MIS on trees as to a problem on \degwei{} trees with bounded degree. Our first goal is to reduce the degree of a tree and still be able to solve MIS. We define the following problem:

\begin{definition}
Consider a tree $T = (V, E)$ where some vertices are ``standard'' vertices and some vertices are ``bypass'' vertices. Let $S\subseteq V$ be a set of vertices such that every standard vertex with a child in $S$ cannot be in $S$ and every bypass vertex is in $S$ if and only if it has a child in $S$. If there exists no vertex set $S'\subseteq V\setminus S$ with at least one standard vertex such that $S\cup S'$ satisfies these two properties, then $S$ is a \textbf{maximal independent set with bypasses}. We denote this problem MISB.
\end{definition}

Viewing this simply as a tree with two types of vertices, this is quite similar to the MIS problem. If we had no bypass vertices, this would be, in fact, MIS. A bypass vertex $b$ in a sense represents its children. If any of its children are in the set and $b$'s parent is standard, then $b$'s parent cannot be in the set. If $b$'s parent is a bypass vertex, then it simply passes this property onto its parent. Additionally, note we require the addition of standard vertices to the set to show it is not independent. Therefore, we do not care about how many bypass vertices are in the set. We will say that the size of an independent set with bypasses is the number of standard vertices in it.

Next, we show that MIS on general trees can be altered to work on trees with degree bounded by $n^\epsilon$ by considering bypass vertices. Note in this theorem when we say ``almost complete $n^\epsilon$-ary tree on $x$ children'', we mean the resulting tree if you greedily filled an $n^\epsilon$-ary tree in a breadth-first manner until it had $x$ children.

For notation, on a tree $T$ with standard vertices $S_T$, we say that $B_A$ for a set of vertices $A \subseteq S_T$ is the set of bypass vertices $b$ such that there exists a descending path from $b$ to some $a\in A$ such that $a$ is the only standard vertex on the path. 

\begin{lemma}\label{lem:MISBtransform}
Consider a tree $T$. There exists a tree $T'=(V',E')$ with max degree bounded by $n^\epsilon$ and set of standard vertices $S_{T'}\subseteq V'$ such that $I'\subseteq V'$ is an MISB on $T'$ implies $I$ is an MIS on $T$ where $I' = I\cup B_I$. Additionally, $|T'| = O(n)$ and $T'$ can be constructed in $O(1/\epsilon)$ rounds with $O(n^\epsilon)$ local memory and $\widetilde{O}(n)$ total space.
\end{lemma}

\begin{proof}
Root $T$ arbitrarily. Transform $T$ into $T'=(V',E')$ as follows: for every $v\in T$ with set of children $Ch_v\subset V$ where $|Ch_v| > n^\epsilon$, replace $v$ with an almost complete $n^\epsilon$-ary tree with $v$ at the root, bypass vertices $\{b_i\}_{i\in[k-1]}$ that are the rest of the internal vertices, and $Ch_v$ as the leaves. Let $I'\subseteq V'$ be any MISB on $T'$ and let $I=I'\cap S_{T'}$ where $S_{T'}$ is the set of standard vertices in $T'$. Obviously, $T'$ is a \degwei{} tree with degree bound $n^\epsilon$, $S_{T'} = V$, and $I' = I\cup B_I$.

Since the size of a tree is within a factor of 2 of its number of leaves and $T'$'s creation clearly creates no leaves, $|T'| = O(|T|) = O(n)$. Additionally, note that the height of each tree replacing high degree vertices is at most $\log_{n^\epsilon}(n) = 1/\epsilon$. In order to implement this transformation, for each node that must be expanded into a tree, partition its children into groups of at most $n^\epsilon$ and put each group on a machine. In each machine, create a bypass node as the parent of all children in the group. Recurse on the newly created bypass nodes until all nodes can fit on one machine, at which point we can link them directly to parent $v$. This requires $O(1/\epsilon)$ rounds, as this is the the height of the tree. Obviously, it satisfies the space constraints.

Assume $I'$ is MISB on $T'$. We show $I$ is an independent set on $T$. Consider any two vertices $u,v\in V$ where $u$ is the child of $v$ and  $u\in I$ (and thus, $u\in I'$). By the construction of $T'$, $u,v\in S_{T'}$, $v$ is an ancestor of $u$ in $T'$, and there is a set of bypass nodes $\{b_i\}_{i\in[k]}$ for some $k$ such that there is a path from $u$ to $v$, $P = (u,b_1,\ldots,b_k,v)$. Since $u\in I'$ and $b_1$ is a bypass vertex, $b_1\in I'$. We can continue this line of reasoning to show $b_k\in I'$. So $v\notin I'$ since $I'$ is independent (with bypass nodes). Therefore, $I$ is an independent set on $T$.

Next, we show $I$ is maximal on $T$. Assume for contradiction there is some $v\in T\setminus I$ such that $I\cup\{v\}$ is an independent set on $T$. Then $v\in S_{T'}$, and $v\notin I'$. Let $T_v$ be the tree $v$ was expanded into, or just the tree of $v$ and its children if it was not expanded. $T_v$'s leaves are $Ch_v$, which are not in $I$ but are all in $S_{T'}$, and thus they cannot be in $I'$. We now show that all bypass vertices $b\in T_v$ are not in $I'$ by inducting on their height. At a height of 1 (i.e., with only leaf-children), $b$'s children must all be in $Ch_v$, and therefore are not in $I'$. Thus, by the rules of bypass vertices in the independent set, $b\notin I'$. For any higher up $b$, given any lower down bypass vertex is not in $I'$, then all of $b$'s children must also not be in $I'$, so $b\notin I'$. This proves that all vertices in $T_v\setminus\{v\}$ are not in $I'$.

Finally, if $v$ has a parent $p$ in $T$, then we know either $p$ is $v$'s parent in $T'$, in which case since $p\notin I$ and $p\in S_{T'}$ then $p\notin I'$, or $v$ is a leaf in the expanded tree of $p$, in which case $v$'s parent is a bypass node. In either case, $p$ cannot interfere with $v$ being put into $I'$. Thus, we have shown that $I'\cup\{v\} \neq I'$ is independent. Thus, $I'$ was not maximal. This is a contradiction, meaning $I$ must be maximal on $T$. 
\end{proof}

Next, we define a dynamic program to solve MISB. Call this dynamic program $D$. It will work bottom-up. For each $v\in V$, $D$ will have one bit to determine if $v$ is in the MISB or not. For each standard leaf, include it in the MISB. For each bypass leaf, do not include it in the MISB. For each standard internal node $v$ whose children have been evaluated, put $v$ in the MISB if none of its children are in the MISB. For each bypass internal node $B$ whose children have been evaluated, put $b$ in the MISB if any of its children are in the MISB. This is a very simple dynamic program that clearly solves MISB.

However, we must now translate this to \degwei{} trees such that we can contract connected components and still solve this problem. We will use a modified dynamic program much like the one from the previous section, simplified, and with bypass vertices accounted for. We will use the function $B(v)$ which is 1 if $v$ is bypass and 0 if $v$ is standard.

\begin{definition}\label{def:dp2}
Let $D$ be the MISB dynamic table on a given tree $T = (V,E)$ that acts on both vertices and edges. Then $D(v)= (c_v, a_v)$ and $D(e) = (\omega_1(e),\omega_2(e))$ where:
\begin{itemize}
\item $c_v$ is a function that computes whether or not $v$ is in the MISB
\item $a_v$ is a constant value
\item $\omega_1(e)$ and $\omega_2(e)$ are binary values on edge $e$
\end{itemize}
and more specifically:
\begin{align*}
c_v =& B(v)a_v\prod_{u\in Ch_v}c_{v,u}+ (1-B(v))\left(1-a_v\prod_{u\in Ch_v} c_{v,u}\right)
\end{align*}
for $c_{v,u} = c_u\omega_1(u,v) + (1-c_u)\omega_2(u,v)$.
\end{definition}

It is not too hard to see that if $a_v = 1$ for all $v\in V$ and $\omega_1(e) = 1$ and $ \omega_2(e) = 0$ for all $e\in E$, this reduces to an implementation of our greedy algorithm.

\begin{proposition}
Consider a tree $T = (V,E)$. The dynamic program $D$ from Definition~\ref{def:dp2} finds the bottom-up greedy MISB when $\omega_1(e) = 1$ and $\omega_2(e)=0$ for all $e\in E$ and $a_v = 1$ for all $v\in V$. In other words, for all $v\in V$, $c_v$ from $D(v)$ indicates whether or not $v$ is in the MISB on $T_v$.
\end{proposition}

Notice, however, that in the greedy algorithm, if $v$ is in the MISB of $T_v$, then it is also in the MISB of $T$. Therefore, our final output has our entire MISB solution, which is in contrast to the MWM solution, where we still had to compute the matching itself.

Consider a connected component $C$ with root $v$. To contract this, we will do a simplified version of the Connected Contraction Process from Definition~\ref{def:contract}. Note in MIS, instead of handling $c_v$ and $c_v'$ we just have a single bit denoting if $v$ is in the MIS or not. For consistency, we call this bit $c_v$. We need to compress $C$ into $C'$ in the same way (i.e., removing internal leaves and contracting maximal chains) and show how to update the values of $D$.

When removing internal leaves $\mathcal{L}$, we can rewrite $c_v$ as:

\[c_v = B(v)a_v\prod_{u\in Ch_v\setminus\mathcal{L}}c_{v,u}\cdot \prod_{\ell\in \mathcal{L}}\ell_{v,u} + (1-B(v))\left(1-a_v\prod_{u\in Ch_v\setminus\mathcal{L}} c_{v,u}\cdot \prod_{\ell\in\mathcal{L}} \ell_{v,u}\right)\]

As in MWM, $\prod_{\ell\in\mathcal{L}} \ell_{v,u}$ is a constant. So we can start with $a_v =1$ and update it as $a_v \gets a_v \cdot \prod_{\ell\in\mathcal{L}} \ell_{v,u}$.

Next, we consider contracting maximal chains. Again, we need to label edges with a tuple-based weight, this time a two-tuple: $D(e) = (\omega_1(e), \omega_2(e))$. Consider our maximal chain $P = (e_1,\ldots, e_k)$ with corresponding vertices $(v_1,\ldots,v_{k+1})$ which we would like to replace with $e_p$. In this problem, $\omega_1(e_P)$ indicates whether or not $v_2$ is in the bottom-up MISB \textit{given} $v_{k+1}$ is in the MISB. On the other hand, $\omega_2(e_P)$ indicates whether or not $v_2$ is in the bottom-up MISB given $v_{k+1}$ is \emph{not} in the MISB. This will help us decide whether or not $v_1$ can be in the MISB. As before, consider child, parent, and grandparent vertices $c$, $p$, and $g$ respectively, such that $c=v_{k+1}$, $p=v_k$, and $g=v_{k-1}$. Assume $c$ is in the greedy MISB. Then $p$ is in the MISB if and only if $w_1(e_k) = 1$ and $p$ is able to be in the MIS according to the rest of the computation on any other children it may have had, which is $c_{p\setminus c}$ as similarly denoted in the previous section. Using this and similar logic, we can update edge weights as follows:

\begin{align*}
\omega_1(e_{k-1},e_k) =& \omega_1(e_k)c_{p\setminus c}
\\\omega_2(e_{k-1},e_k) =& \omega_2(e_k)c_{p\setminus c}
\end{align*}

As before, these are all known constants, and we can repeatedly apply this to find $\omega_1(e_P)$ and $\omega_2(e_P)$. Thus, we have shown as before that we can contract $C$ into a tree $C'$ such that $|C'| = O(\deg(C))$ using  this simplified Connected Contraction Process. As before:

\begin{proposition}\label{prop:dp2}
This simplified Connected Contraction Process replaces any connected component $C$ with a component $C'$ such that $|C'| = O(\deg(C))$ and $c_v$ remains the same in $C$ and $C'$.
\end{proposition}

Now we introduce the degree-weighted problem. As before, we let edge weights be stored by the associated child endpoint.

\begin{definition}\label{def:dwmisb}
Consider a \degwei{} tree $T = (V,E,W)$ where we have $W(v) = (D(v), D(v,p))$ for all $v\in V$ with parent $p\in V$ and $W(r) = (D(r))$ for root $r$, where all $D$ values are stored as explained in the simplified Connected Contraction Process. Then \textbf{maximum independent set with bypass vertices} on $T$ is equivalent to the problem of solving $D(v)$ for all $v\in V$.
\end{definition}

By Proposition~\ref{prop:dp2} and Lemma~\ref{lem:MISBtransform}, solving MISB on \degwei{} trees can be used to solve MIS on standard trees with the aforementioned input $W$. Now we apply the main algorithms.

\begin{proof}[Proof of Lemma~\ref{lem:ampcmis}]
All we need to do is introduce a connected contracting function $\C$. For any component $C\subseteq V$, to compute $\C(C)$, apply the simplified Connected Contraction Process to contract the component and get $W(C)$. By Proposition~\ref{prop:dp2}, this ensures that $\dim(W(C)) = \widetilde{O}(\deg(v))$, thus this is a valid weight vector. It is also a valid connected contracting function since $c_v$ does not change as a function of $Ch_C$. Therefore, by Theorem~\ref{thm:boundedmain}, we can solve MISB on \degwei{} trees in $O(1/\epsilon^2)$ AMPC rounds with $O(n^\epsilon)$ memory per machine and $\widetilde{O}(n)$ total memory. By the efficient transformation from Lemma~\ref{lem:MISBtransform}, this can be used to solve MIS on normal trees with the same complexities.
\end{proof}

A few simple modifications yields the same result for maximal matching.

\begin{lemma}\label{lem:maxlmatch}
Maximal matching on trees can be solved in $O(1/\epsilon^2)$ AMPC rounds with $O(n^\epsilon)$ local memory and $\widetilde{O}(n)$ total memory.
\end{lemma}